\documentclass[11pt]{article}
\usepackage[margin=0.8in]{geometry}

\usepackage{varioref}
\usepackage[usenames,svgnames,xcdraw,table]{xcolor}
\definecolor{DarkBlue}{rgb}{0.1,0.1,0.5}
\definecolor{DarkGreen}{rgb}{0.1,0.5,0.1}

\usepackage{hyperref}
\hypersetup{
   colorlinks   = true,
 linkcolor    = DarkBlue, 
 urlcolor     = DarkBlue, 
	 citecolor    = DarkGreen 
}

\usepackage{appendix}

\usepackage{mathtools}

\usepackage{algorithmic}

\usepackage{algorithm}
\usepackage{array}

\usepackage[font={small,it}]{caption}

\usepackage{graphicx,epsfig,amsmath,latexsym,amssymb,verbatim}

\usepackage[square,sort,semicolon,numbers]{natbib}

\newcommand{\extra}[1]{}

\usepackage{amsmath,amssymb,amsfonts}
\usepackage{amsthm}
\usepackage{thmtools,thm-restate}
\usepackage{enumitem}
\usepackage{nicefrac}

\newtheorem{theorem}{Theorem}

\def\squareforqed{\hbox{\rlap{$\sqcap$}$\sqcup$}}
\def\qed{\ifmmode\squareforqed\else{\unskip\nobreak\hfil
\penalty50\hskip1em\null\nobreak\hfil\squareforqed
\parfillskip=0pt\finalhyphendemerits=0\endgraf}\fi}
\def\endenv{\ifmmode\;\else{\unskip\nobreak\hfil
\penalty50\hskip1em\null\nobreak\hfil\;
\parfillskip=0pt\finalhyphendemerits=0\endgraf}\fi}

\renewenvironment{proof}{\noindent \textbf{{Proof~} }}{\qed\medskip}
\newenvironment{proof+}[1]{\noindent \textbf{{Proof #1~} }}{\qed\medskip}

\mathchardef\ordinarycolon\mathcode`\:
\mathcode`\:=\string"8000
\def\vcentcolon{\mathrel{\mathop\ordinarycolon}}
\begingroup \catcode`\:=\active
  \lowercase{\endgroup
  \let :\vcentcolon
  }

\newcommand{\nc}{\newcommand}

\nc{\barA}{\overline{A}}
\nc{\barB}{\overline{B}}
\nc{\barC}{\overline{C}}
\nc{\barD}{\overline{D}}
\nc{\barR}{\overline{R}}
\nc{\barX}{\overline{X}}
\nc{\barY}{\overline{Y}}
\nc{\barU}{\overline{U}}

\usepackage[utf8]{inputenc}
\usepackage{amsmath}
\usepackage{amssymb,amsthm,mathtools}
\usepackage{amsmath,amsfonts,amsthm} 
\usepackage{caption} 
\usepackage{xcolor}
\usepackage{hyperref}
\usepackage[all]{hypcap}

\usepackage[font=small,labelfont=bf]{caption}
\usepackage{xspace}
\usepackage{etoolbox}

\newcommand{\allo}[3]{\left ({#1}^{#2}_{#3}, \ldots #1^{#2}_n \right )}

\newcommand{\I}[2]{\langle #1,#2,v \rangle}
\newcommand{\Mp}[3][]{{\rm M}_p(#2^{#1}_{#3},\ldots , #2^{#1}_n)}

\newcommand{\Pallo}{(\{g_1\},\ldots,\{g_k\},B_{k+1},\ldots , B_{n})}
\newcommand{\summi}[3]{\sum \limits _{i=#1}^{#2} v(#3_i)}

\newtheorem {Proposition}{Proposition}
\newtheorem {Lemma}{Lemma}

\renewcommand{\footnotesize}{\small}

\title{\bfseries Uniform Welfare Guarantees Under \\ Identical Subadditive Valuations}
\author{Siddharth Barman\thanks{Indian Institute of Science. {\tt barman@iisc.ac.in}} \and Ranjani G.~Sundaram\thanks{Chennai Mathematical Institute. {\tt ranjanigs@cmi.ac.in}}}

\date{}

\begin{document}
\maketitle

\begin{abstract}
We study the problem of allocating indivisible goods among agents that have an identical subadditive valuation over the goods. The extent of fairness and efficiency of allocations is measured by the \emph{generalized means} of the values that the allocations generate among the agents. Parameterized by an exponent term $p$, generalized-mean welfares encompass multiple well-studied objectives, such as social welfare, Nash social welfare, and egalitarian welfare. 

We establish that, under identical subadditive valuations and in the demand oracle model, one can efficiently find a single allocation that approximates the optimal generalized-mean welfare---to within a factor of $40$---uniformly for all $p \in (-\infty, 1]$. Hence, by way of a constant-factor approximation algorithm, we obtain novel results for maximizing Nash social welfare and egalitarian welfare for identical subadditive valuations.
\end{abstract}

\section{Introduction}

A significant body of recent work, in algorithmic game theory, has been directed towards the study of fair and efficient allocation of indivisible goods among agents; see, e.g.,~\cite{endriss2017trends} and~\cite{brandt2016handbook}. This thread of research has led to the development of multiple algorithms and platforms (e.g., {Spliddit}~\cite{goldman2015spliddit}) which, in particular, address settings wherein discrete resources (that cannot be fractionally allocated) need to be partitioned among multiple agents. Contributing to this line of work, the current paper studies discrete fair division from a welfarist perspective. 

We specifically address the problem of finding allocations (of indivisible goods) that (approximately) maximize the \emph{generalized means} of the agents' valuations. Formally, for exponent parameter $p \in \mathbb{R}$, the $p${th} generalized mean, of $n$ nonnegative values $\{v_i\}_{i=1}^n$, is defined as ${\rm M}_p(v_1, \ldots, v_n) \coloneqq \left( \frac{1}{n} \sum_i v_i^p \right)^{\frac{1}{p}}$. Parameterized by $p$, this family of functions includes well-studied fairness and efficiency objectives, such as average social welfare ($p=1$), Nash social welfare ($p \to 0$), and egalitarian welfare ($p \to -\infty$). In fact, generalized means---with the exponent parameter $p$ in the range $(-\infty, 1]$---admit a fundamental axiomatic characterization: up to monotonic transformations, generalized means (with $p \in (-\infty, 1]$) exactly constitute \emph{the} family of welfare functions that satisfy the \emph{Pigou-Dalton transfer principle} and a few other key axioms~\cite{moulin2004fair}.\footnote{Note that generalized means are ordinally equivalent to CES (constant elasticity of substitution) functions.} Hence, by way of developing a single approximation algorithm for maximizing generalized means, the current work provides a unified treatment of multiple fairness and efficiency measures. 

With generalized mean as our objective, we focus on fair-division instances in which the agents have a common \emph{subadditive} (i.e., complement free) valuation. Formally, a set function $v$, defined over a set of indivisible goods $[m]$, is a said to be subadditive iff, for all subsets $A$ and $B$ of $[m]$, we have $v(A \cup B) \leq v(A) + v(B)$. This class of functions includes many other well-studied valuation families, namely \emph{XOS}, \emph{submodular}, and \emph{additive} valuations.\footnote{Recall that a submodular function $f$ is defined by a diminishing returns property: $f(A + e) - f(A) \geq f(B + e) - f(B)$, for all subsets $A \subseteq B$ and $e \notin B$.} These function classes have been used extensively in computer science and mathematical economics to represent agents' valuations. Of particular relevance here are results that (in the context of combinatorial auctions) address the problem of maximizing social welfare under submodular, XOS, and, more generally,  subadditive valuations~\cite{nisan2007algorithmic}. 


The focus on a common valuation function across the agents provides a technically interesting and applicable subclass of fair-division problems--as a stylized application, consider a setting in which the agents' values represent money, i.e., for every agent, the value of each subset (of the goods) is equal to the subset's monetary worth. Here, one encounters subadditivity when considering goods that are substitutes of each other. Also, from a technical standpoint, we note that the problem of maximizing social welfare is {\rm APX}-hard even under identical submodular~\cite{khot2008inapproximability} and subadditive valuations~\cite{dobzinski2005approximation}. Appendix \ref{APX_Hardness} extends this hardness result to all $p \in (-\infty, 1]$.\\

\noindent
{\bf Our Results:} Addressing fair-division instances with identical subadditive valuations, we develop an efficient constant-factor approximation algorithm for the generalized-mean objective (Theorem~\ref{MainTheorem}). Specifically, our algorithm computes an allocation (of the indivisible goods among the agents), $\mathcal{A}$, with the property that its generalized-mean welfare, ${\rm M}_p (\mathcal{A})$, is at least ${1}/{40}$ times the optimal $p$-mean welfare, for all $p \in (-\infty, 1]$. This result in fact implies an interesting existential guarantee as well: if in a fair-division instance the agents' valuations are identical and subadditive, then there exists a single allocation that uniformly approximates the optimal $p$-mean welfare for all $p \in (-\infty, 1]$.  

The tradeoff between fairness and economic efficiency is an important consideration in fair division literature.\footnote{For example, consider the work on price of fairness~\cite{bertsimas2011price,bei2019price}} The relevance of the above-mentioned existential guarantee is substantiated by the fact that this result reasonably mitigates the fairness-efficiency tradeoff in the current context; it shows that for identical subadditive valuations there exists a single allocation which is near optimal with respect to efficiency objectives (in particular, social welfare) as well as fairness measures (e.g., egalitarian welfare). Note that such an allocation cannot be simply obtained by selecting an arbitrary partition that (approximately) maximizes social welfare: under identical additive valuations, all the allocations have the same social welfare, even the ones with egalitarian welfare equal to zero. One can also construct instances, with identical subadditive valuations, wherein particular allocations have optimal egalitarian welfare, but subpar social welfare. 

Even specific instantiations of our algorithmic guarantee provide novel results: while the problem of maximizing Nash social welfare, among $n$ agents, admits an $\mathcal{O}(n \log n)$-approximation under nonidentical submodular valuations~\cite{garg2020approximating}, the current work provides a novel (constant-factor) approximation guarantee for maximizing Nash social welfare when the agents share a common subadditive (and, hence, submodular) valuation.\footnote{Under nonidentical additive valuations, there exists a polynomial-time $1.45$-approximation algorithm for maximizing Nash social welfare~\cite{barman2018finding}. Furthermore, under identical additive valuations, maximizing Nash social welfare admits a polynomial-time approximation scheme~\cite{nguyen2014minimizing,barman2018greedy}.}  Analogously, the instantiation of our result for egalitarian welfare is interesting in and of itself.

Given that the valuations considered in this work express combinatorial preferences, a naive representation of such set functions would require exponential (in the number of goods) values, one for each subset of the goods. Hence, to primarily focus on the underlying computational aspects and not on the representation details, much of prior work assumes that the valuations are provided via oracles that can only answer  particular type of queries. The most basic oracle considered in literature answers \emph{value queries}: given a subset of the indivisible goods, the value oracle returns the value of this subset. In this value oracle model, the work of Vondr\'{a}k~\cite{vondrak2008optimal} considers submodular valuations and provides an efficient $\frac{e}{e-1}$-approximation algorithm for maximizing social welfare. Using this method as a subroutine and, hence, completely in the value oracle model, our algorithm achieves the above-mentioned approximation guarantee for identical submodular valuations.

Another well-studied oracle addresses \emph{demand queries}. Specifically, such an oracle, when queried with an assignment of prices $p_1, \ldots, p_m \in \mathbb{R}$ to the $m$ goods, returns $\max_{S \subseteq [m]} \left( v(S) - \sum_{j \in S} p_j \right)$, for the underlying valuation function $v$.\footnote{Observe that a value query can be simulated via polynomially many demand queries. Though, the converse is not true~\cite{nisan2007algorithmic}.} Demand oracles have been often utilized in prior work for addressing social welfare maximization in the context of subadditive and XOS valuations~\cite{nisan2007algorithmic}. In particular, the work of Fiege~\cite{feige2009maximizing} shows that, under subadditive valuations and assuming oracle access to {demand queries},\footnote{This result holds even if the agents have distinct, but subadditive, valuations.} the social welfare maximization problem admits an efficient $2$-approximation algorithm. Demand queries are unavoidable in the subadditive case: one can directly extend the result of Dobzinski et al.~\cite{DobzinskiNS10} to show that, even under identical (subadditive) valuations, any sub-linear (in $n$) approximation of the optimal social welfare requires exponentially many value queries. At the same time, we note that our algorithm requires demand oracle access \emph{only} to implement the $2$-approximation algorithm of Fiege~\cite{feige2009maximizing} as a subroutine. Beyond this, we can work with the value oracle.  \\

\noindent
{\bf Related Work:} Multiple algorithmic and hardness results have been developed to address welfare maximization in the context of indivisible goods/discrete resources. Though, in contrast to the present paper, prior work in this direction has primarily addressed one welfare function at a time. 

As mentioned previously, maximizing social welfare and Nash social welfare (see, e.g., \cite{cole2018approximating} and references therein) has been actively studied in algorithmic game theory. Egalitarian welfare has also been addressed in prior work--this welfare maximization problem is also referred to as the  max-min allocation problem (or the Santa Claus problem); see, e.g.,~\cite{DBLP:journals/corr/AnnamalaiKS14}. Specifically, for maximizing egalitarian welfare under additive and nonidentical valuations, the result of Chakrabarty et al.~\cite{chakrabarty2009allocating} provides an $\widetilde{\mathcal{O}}(n^{\varepsilon})$-approximation algorithm that runs in time $\mathcal{O}(n^\frac{1}{\varepsilon})$; here $n$ denotes the number of agents and $\varepsilon >0$.  Furthermore, under nonidentical submodular valuations, the problem of maximizing egalitarian welfare is known to admit a polynomial-time $\widetilde{\mathcal{O}}(n^{1/4} m^{1/2})$-approximation algorithm~\cite{goemans2009approximating}; here $m$ is the number of goods. In contrast to these sublinear approximations, this paper shows that, if the agents' valuations are identical, then even under subadditive valuations the problem of maximizing egalitarian welfare admits a constant-factor approximation guarantee.



\section{Notation and Preliminaries}

An instance of a fair-division problem corresponds to a tuple $\langle [m], [n], v \rangle$, where $[m]= \left\{1,2,\ldots, m \right\}$ denotes the set of $m \in \mathbb{N}$ indivisible {goods} that have to be allocated (partitioned) among the set of $n \in \mathbb{N}$ agents, $[n]=\{1, 2, \ldots, n\}$. 
Here, $v: 2^{[m]} \mapsto \mathbb{R}_+$ represents the (identical) valuation function of the agents;\footnote{Recall that this work addresses fair-division instances in which all the agents have a common valuation function.} specifically, $v(S) \in \mathbb{R}_+$ is the value that each agent $i \in [n]$ has for a subset of goods $S \subseteq [m]$. 

We will assume throughout that the valuation function $v$ is (i) normalized: $v(\emptyset) = 0$, (ii) monotone:  $v(A) \leq v(B)$ for all $A \subseteq B \subseteq [m]$, and (iii) {subadditive}: $v(A \cup B) \leq v(A) + v(B)$ for all subsets $A, B \subseteq [m]$.

Write $\Pi_n([m])$ to denote the collection of all $n$ partitions of the indivisible goods $[m]$. We use the term \textit{allocation} to refer to an $n$-partition $\mathcal{A} = \allo{A}{}{1} \in \Pi_n([m])$ of the $m$ goods. Here, $A_i$ denotes the subset of goods allocated to agent $i \in [n]$ and will be referred to as a \emph{bundle}. 

Generalized (H\"{o}lder) means, ${\rm M}_p$, constitute a family of functions that capture multiple fairness and efficiency measures. Formally, for an exponent parameter $p \in \mathbb{R}$, the $p${th} generalized mean of $n$ nonnegative numbers $x_1,\ldots , x_n \in \mathbb{R}_+$ is defined as $\Mp{x}{1} \coloneqq \left( \frac{1}{n} \sum \limits _{i=1}^n x_i^p \right )^\frac{1}{p}$.

Note that, when $p=1$, ${\rm M}_p$ reduces to the arithmetic mean. Also, as $p$ tends to zero, ${\rm M}_p$, in the limit, is equal to the geometric mean and $\lim_{p \rightarrow -\infty}  \Mp{x}{1} = \min\{x_1, x_2, \ldots, x_n\}$. Hence, following standard convention, we will write ${\rm M}_0(x_1, \ldots, x_n) = \left(\prod_{i=1}^n x_i \right)^{1/n}$ and ${\rm M}_{-\infty}(x_1, \ldots, x_n) = \min_i x_i $.

Considering generalized means as a parameterized collection of welfare objectives, we define the \emph{$p$-mean welfare}, ${\rm M}_p(\mathcal{A})$, of an allocation $\mathcal{A}=(A_1, A_2, \ldots, A_n)$ as  \begin{align}
{\rm M}_p(\mathcal{A}) & \coloneqq  {\rm M}_p\left( v(A_1), \ldots, v(A_n) \right) = \left( \frac{1}{n} \sum_{i=1}^n v (A_i)^p \right)^{1/p} \label{eq:generalized-mean}
\end{align}
Here, $v$ is the (common) valuation function of the agents. Indeed, with $p$ equal to one, zero, and $-\infty$, the $p$-mean welfare, respectively,  corresponds to (average) social welfare, Nash social welfare, and egalitarian welfare.

Given a fair-division instance $\mathcal{I}=\langle [m], [n], v \rangle$ and $p \in (-\infty, 1]$, ideally, we would like to find an allocation $\mathcal{A} = (A_1, \ldots, A_n)$ with as large an ${\rm M}_p(\mathcal{A})$ value as possible, i.e., maximize the $p$-mean welfare. An allocation that achieves this goal will be referred to as a \emph{$p$-optimal allocation} and denoted by $\mathcal{A}^*(\mathcal{I}, p)=(A^*_1(\mathcal{I}, p), A^*_2(\mathcal{I}, p), \ldots, A^*_n(\mathcal{I}, p))$. 


We note that, under identical, subadditive valuations, finding a $p$-optimal allocation is {\rm APX}-hard, for any $p \in (-\infty, 1]$ (Appendix~\ref{APX_Hardness}). Hence, the current work considers approximation guarantees. In particular, for fair-division instances $\mathcal{I}$ in which the agents have a common subadditive valuation, we develop a polynomial-time algorithm that computes a single allocation $\mathcal{A}$ with the property that ${\rm M}_p(\mathcal{A})\geq \frac{1}{40} {\rm M}_p({\mathcal{A}^{*}}(\mathcal{I}, p))$ for all $p \in (-\infty, 1]$. That is, the developed algorithm achieves an approximation ratio of $40$ uniformly for all $p \in (-\infty, 1]$. 


The work of Fiege~\cite{feige2009maximizing} shows that, for subadditive valuations, the social-welfare maximization problem (equivalently, the problem of maximizing ${\rm M}_1(\cdot)$) admits an efficient $2$-approximation algorithm, assuming oracle access to {demand queries}. In particular, such an oracle, when queried with an assignment of prices $p_1, \ldots, p_m \in \mathbb{R}$ to the $m$ goods, returns $\max_{S \subseteq [m]} \left( v(S) - \sum_{j \in S} p_j \right)$. Our algorithm requires demand oracle access {only} to implement the $2$-approximation algorithm of Fiege~\cite{feige2009maximizing} as a subroutine. Beyond this, we can work with the basic value oracle, which when queried with a subset of goods $S \subseteq [m]$, returns $v(S)$. 

In fact, if the underlying valuation is submodular, then one can invoke the result of Vondr\'{a}k~\cite{vondrak2008optimal} (instead of using the approximation algorithm by Feige~\cite{feige2009maximizing}) and efficiently obtain a $\frac{e}{e-1}$-approximation for the social-welfare maximization problem in the value oracle model. Hence, under a submodular valuation, our algorithm can be implemented entirely in the standard value oracle model. 

For a fair-division instance $\mathcal{I}$, write ${\rm F}(\mathcal{I})$ to denote the $1$-mean welfare ${\rm M}_1$ (i.e., the average social welfare) of the allocation computed by the approximation algorithm of Feige~\cite{feige2009maximizing}. The approximation guarantee established in~\cite{feige2009maximizing} ensures that---for any instance $\mathcal{I}$ with a subadditive valuation---we have ${\rm F}(\mathcal{I}) \geq \frac{1}{2} {\textrm M}_1(\mathcal{A}^{*}(\mathcal{I}, 1))$. Here, $\mathcal{A}^{*}(\mathcal{I}, 1)$ denotes a $1$-optimal allocation, i.e., it maximizes the (average) social welfare in $\mathcal{I}$.

\section{Maximizing $p$-Mean Welfare}

Addressing fair-division instances with identical subadditive valuations, this section presents an efficient algorithm for computing a constant-factor approximation to the $p$-mean welfare objective, uniformly for all $p \in (-\infty, 1]$. 

The algorithm consists of two phases, Algorithm~\ref{Alg} (\textsc{Alg}) and Algorithm~\ref{AlgSub} (\textsc{AlgLow}). In the first phase, ``high-value'' goods are assigned as singletons--we use the approximation algorithm of Feige~\cite{feige2009maximizing} to obtain an estimate of the optimal  $1$-mean welfare and deem a good to be of high value if its valuation is at least a constant (specifically, $\nicefrac{1}{3.53}$) times this estimate. Intuitively, the estimate provides a useful benchmark, since the optimal $1$-mean welfare upper bounds the optimal $p$-mean welfare for all $p \in (-\infty, 1]$ (Proposition~\ref{Monotonicity}); this bound essentially follows from the generalized mean inequality~\cite{bullen1988mathematics} which asserts that, for all $p \in (-\infty, 1]$, the $p$-mean welfare of any allocation $\mathcal{A}$ is at most its $1$-mean welfare, ${\textrm M}_p(\mathcal{A}) \leq {\textrm M}_1(\mathcal{A})$. 

Therefore, in phase one of the algorithm, we sort the goods in non-increasing order by value and iteratively select goods, which by themselves provide a value comparable to that of the optimal $p$-mean welfare. In each iteration, the selected good is assigned as a singleton to an agent and this agent-good pair is removed from consideration. Note that such an update leads to a new fair-division instance with one less good and one less agent, as well as a potentially different optimal $1$-mean welfare. The key technical issue here is that the change in the optimal $1$-mean welfare (and, hence, its estimate obtained via Feige's  algorithm) can be non-monotonic. Nonetheless, via an inductive argument, we show that the welfare contribution of the goods assigned (as singletons) in the first phase is sufficiently large (Lemma~\ref{Induction_argument}). 

The first phase terminates when we obtain an instance $\mathcal{J}$ wherein each good is of value no more than a constant times its optimal $1$-mean welfare. The second phase (\textsc{AlgLow}) is designed to address such a fair-division instance. In particular, we show that, in the absence of high-value goods, we can efficiently find an allocation $\mathcal{B}=(B_i)_i$ such that each bundle $B_i$ is of value at least constant times the optimal $p$-mean welfare of $\mathcal{J}$. To obtain the allocation $\mathcal{B}$, we first compute (via Feige's approximation algorithm) an allocation $\mathcal{S}=(S_j)_j$ that provides a $2$-approximation to the optimal $1$-mean welfare of $\mathcal{J}$. Subsequently, we show that the subsets $S_j$s, that have appropriately high value, can be  partitioned to form the desired bundles $B_i$s, which constitute the allocation $\mathcal{B}$. 

Multiple technical lemmas (in Sections~\ref{Supporting_Lemmas} and~\ref{section:stitching-lemma}) are required to show that  the two phases in combination lead to the desired $p$-welfare bound. It is also relevant to note that, while the above-mentioned ideas hold at a high level, the formal guarantees are obtained by separately analyzing different ranges of the exponent parameter $p$. 




\floatname{algorithm}{Algorithm}
\begin{algorithm}[ht]
  \caption{\textsc{Alg}} \label{Alg}
  \textbf{Input:} A fair-division instance $\mathcal{I}= \langle [m],[n],v \rangle$ with demand oracle access to the subadditive valuation function $v$. \\
  \textbf{Output:} An allocation $\mathcal{A} = (A_1, A_2, \ldots, A_n)$ 
  \begin{algorithmic}[1]
        \STATE Initialize the set of agents ${U}=[n]$, the set of goods $G=[m]$, and bundle $A_i= \emptyset $ for all $i\in {U}$
        \STATE Index all the goods in non-increasing order of value $v(g_1) \geq v(g_2) \geq  \ldots \geq v(g_m)$\label{Ordered_Goods}
        \STATE Set $\mathcal{I}^0= \I{G}{{U}}$ and initialize $t=1$    \COMMENT{Recall that ${\rm F}(\mathcal{I}) = {\rm M}_1 (\mathcal{S})$, where $\mathcal{S}$ denotes the allocation obtained by executing Feige's algorithm~\cite{feige2009maximizing} on instance $\mathcal{I}$.}
        \WHILE {  $v(g_t) \geq \frac{1}{3.53} \ {\rm F}(\mathcal{I}^{t-1})$ } \label{Threshold}
         \STATE Allocate $A_t \leftarrow  \{ g_t \}$ and update $G \leftarrow G\setminus \{g_t\}$ along with ${U} \leftarrow {U} \setminus \{ t \}$
                   \STATE Set $\mathcal{I}^{t}=\I{G}{{U}}$ and update $t \leftarrow t+1$
        \ENDWHILE \label{step:end-while}
        \STATE Set $(A_{t}, A_{t+1},\ldots, A_n) = \textsc{AlgLow}(G, {U} ,v)$ \COMMENT{This step corresponds to the second phase of the algorithm which assigns bundles to the remaining $| {U}| = n-t +1$ agents. Also, note that, in the current instance $\mathcal{J} \coloneqq \langle G, {U}, v \rangle$, for every good $g \in G$ we have $v(g) < \frac{1}{3.53} {\rm F} (\mathcal{J})$.}
        \RETURN allocation $\mathcal{A} = (A_1,A_2,...,A_n).$
     
  \end{algorithmic}
\end{algorithm}


\floatname{algorithm}{Algorithm}
\begin{algorithm}[h]
  \caption{\textsc{AlgLow} } \label{AlgSub}
   \textbf{Input:} A fair-division instance $\mathcal{J} = \langle G, U,v \rangle$ with demand oracle access to the subadditive valuation function $v$. \\\textbf{Output:} An allocation $\mathcal{B} = (B_1, B_2, \ldots, B_{|U|})$ 
  \footnotesize
  \begin{algorithmic} [1]
        \STATE Execute Feige's approximation algorithm~\cite{feige2009maximizing} on the given instance $\mathcal{J}$ to compute allocation $\mathcal{S} = (S_1, S_2, \ldots, S_{|U|})$.  
        \COMMENT{Note that allocation $\mathcal{S}$ provides a $2$-approximation to the optimal $1$-mean welfare of $\mathcal{J}$, ${\rm M}_1 (\mathcal{S}) = {\rm F}(\mathcal{J}) \geq \frac{1}{2} {\rm M}_1 \left(\mathcal{A}^*(\mathcal{J}, 1) \right)$}
        \STATE Index the bundles such that $v(S_1) \geq \ldots \geq v(S_{|U|})$ and initialize $i=a = 1$ along with $B_\ell = \emptyset$ for $1 \leq \ell \leq |U|$ \\ \COMMENT{Lemma~\ref{Low_valued} shows that the following loop runs to completion} 
         \WHILE {agent index $a < |U|$}
         \STATE Consider an arbitrary good $g \in S_i$   
         \IF {$v( B_a \cup \{ g\}) < \frac{1}{3} {\rm F} (\mathcal{J}) $ } 
            \STATE Update $B_a \leftarrow B_a \cup \{ g \}$ and $S_i \leftarrow S_i \setminus \{ g\}$ \COMMENT{Here good $g$ is assigned to bundle $B_a$ to  increase its value}
        \ELSE
        \STATE Update $a \leftarrow a + 1$ \COMMENT{This update is performed when sufficient value has been accumulated in a bundle}
        \ENDIF
        \IF{$v(S_i) < \frac{1}{3} {\rm F} (\mathcal{J}) $}
        \STATE Update $i \leftarrow i + 1$ \COMMENT{Once the value of $S_i$ drops below $\frac{1}{3} {\rm F}(\mathcal{J})$ we consider the next bundle in $\mathcal{S}$}
        \ENDIF
\ENDWHILE
\STATE $B_{|U|} \leftarrow B_{|U|} \cup \left( G \setminus (\bigcup \limits _{a=1}^{|U|-1} B_a )\right)$ \COMMENT{Assign the remaining elements to $B_{|U|}$} 
        
        \RETURN partition $\mathcal{B} = (B_1,\ldots ,B_{|U|}).$
     
  \end{algorithmic}
\end{algorithm}



The following theorem constitutes the main result of the current work. It asserts that Algorithm~\ref{Alg} (\textsc{Alg}) achieves a constant-factor  approximation ratio for the $p$-mean welfare maximization problem. 

\newtheorem{thm}{Theorem}[section]
\newcommand{\thistheoremname}{Main Theorem }

\newtheorem*{genericthm*}{\thistheoremname}
\newenvironment{namedthm*}[1]
  {\renewcommand{\thistheoremname}{#1}%
   \begin{genericthm*}}
  {\end{genericthm*}}

\begin{theorem} [Main Result] \label{MainTheorem}
Let $\mathcal{I} = \langle [m], [n], v \rangle$ be a fair-division instance wherein all the agents have an identical, subadditive valuation function $v$. Given  demand oracle access to $v$, \textsc{Alg} computes in polynomial time an allocation $\mathcal{A}$ that, for all $p \in (-\infty, 1]$, provides a $40$-approximation to the optimal $p$-mean welfare, i.e., ${\rm M}_p(\mathcal{A})\geq \frac{1}{40} \ {\rm M}_p(\mathcal{A}^*(\mathcal{I},p))$, for all $p \in (-\infty, 1]$; here, $\mathcal{A}^*(\mathcal{I},p)$ is the $p$-optimal allocation in $\mathcal{I}$. 
\end{theorem}

 
 We first consider an instance $\mathcal{J}$ wherein all the goods are of value a constant times less than ${\rm F}(\mathcal{J})$ and prove that, for such an instance, \textsc{AlgLow} finds an allocation in which the value of every bundle is comparable to the optimal average social welfare of $\mathcal{J}$. In Section \ref{subsection:p-inf-half}, we use this fact and supporting lemmas from Sections~\ref{Supporting_Lemmas} and~\ref{section:stitching-lemma} to prove Theorem~\ref{MainTheorem} for $p\in (-\infty,0.4)$. Finally, in Section \ref{subsection:p-half-one}, we prove the main result for $p\in [0.4,1].$ 
 
We start with the following observation to upper bound the optimal $p$-mean welfare in terms of the optimal $1$-mean welfare. 
\begin{Proposition}\label{Monotonicity}
Let $\mathcal{I}$ be a fair-division instance in which all the agents have an identical, subadditive valuation $v$. Then, for each $p\in(-\infty,1],$ the optimal $1$-mean welfare is at least as large as the optimal $p$-mean welfare:
\begin{align*}
{\rm M}_1(\mathcal{A}^{*}(\mathcal{I},1)) \geq {\rm M}_p(\mathcal{A}^{*}(\mathcal{I},p)) \  \ \text{for every } \; p \in \left(- \infty , 1\right]
\end{align*}  
\end{Proposition}
\begin{proof}
The generalized mean inequality (see, e.g., ~\cite{bullen1988mathematics}) applied to the allocation $\mathcal{A}^{*}(\mathcal{I},p)$, gives us ${\rm M}_p(\mathcal{A}^{*}(\mathcal{I},p)) \leq {\rm M}_1 (\mathcal{A}^{*}(\mathcal{I},p))$, for all $p \in (-\infty, 1]$. By definition, the allocation $\mathcal{A}^*(\mathcal{I}, 1)$ maximizes the $1$-mean welfare, ${\rm M}_1(\cdot)$, and, hence, the claim follows ${\rm M}_1 (\mathcal{A}^{*}(\mathcal{I},1)) \geq {\rm M}_1 (\mathcal{A}^{*}(\mathcal{I},p))  \geq {\rm M}_p(\mathcal{A}^{*}(\mathcal{I},p))$. 
\end{proof}

\section{Approximation Guarantee for \textsc{AlgLow}}
\label{Approx_for_AlgSub}
This section addresses the second phase of the algorithm (\textsc{AlgLow}) that---by the processing performed in the while-loop of \textsc{Alg}---solely needs to  consider fair-division instances $ \mathcal{J}=\I{G}{U}$ wherein all the goods $g \in G$ satisfy $v(g)\leq \frac{1}{3.53}{\rm F}(\mathcal{J})$, i.e., the goods are of ``low value.''  The following lemma establishes that, for such instances, \textsc{AlgLow} finds bundles each with value comparable to the optimal $1$-mean welfare (and, hence, comparable to the optimal $p$-mean welfare) of $\mathcal{J}$. 

Recall that here $G$ is a subset of the original set of goods $[m]$, $U$ is a subset of the $[n]$ agents, and ${\rm F}(\mathcal{J})$ denotes the $1$-mean welfare (average social welfare) of the allocation computed by Feige's approximation algorithm for instance $\mathcal{J}$.

\begin{Lemma}\label{Low_valued}
Let $\mathcal{J} = \I{G}{U}$ be a fair-division instance in which all the agents have an identical subadditive valuation function $v$, and every good $g\in G$ satisfies   
$v(g)\leq \frac{1}{3.53}{\rm F}(\mathcal{J})$. Then, in the demand oracle model, the algorithm \textsc{AlgLow} efficiently computes an allocation  $\mathcal{B}=(B_1,\ldots , B_{|U|})$ with the property that, for all $i \in \{1, \ldots, |U|\}$,
\begin{align*}
v(B_i)\geq \frac{1}{40} {\rm M}_1(\mathcal{A}^{*}(\mathcal{J},1)) \geq \frac{1}{40} {\rm M}_p(\mathcal{A}^{*}(\mathcal{J},p)).
\end{align*}
\end{Lemma}
\begin{proof}
For input instance $\mathcal{J}$, Feige's algorithm returns an allocation $\mathcal{S}=(S_1,S_2,\ldots,S_{|U|})$ with near-optimal average social welfare:
\begin{align*}
\frac{1}{|U|}\summi{1}{|U|}{S} & = {\rm F}(\mathcal{J}) \geq \frac{1}{2}{\rm M}_1(\mathcal{A}^{*}(\mathcal{J},1))
\end{align*}

Given allocation $\mathcal{S}=(S_1,S_2,\ldots,S_{|U|})$, we show that one can partition $S_i$s to form $|U|$ bundles such that the value of each  bundle is at least $\frac{1}{40} {\rm M}_1(\mathcal{A}^{*}(\mathcal{J},1))$. Note that, by the assumption in the lemma, for each $g \in G$ we have 
\begin{align}\label{2}
v(g)\leq \frac{1}{3.53}{\rm F}(\mathcal{J})
\end{align} 

Let $H \coloneqq \left\{ i \in \{1, 2, \ldots, |U|\} \mid v(S_i) \geq \frac{1}{3} {\rm F}(\mathcal{J}) \right\}$ denote the subsets in $\mathcal{S}$ with value at least $\frac{1}{3} {\rm F}(\mathcal{J})$. These are the subsets we split to form bundles $B_i's$ that form the output allocation $\mathcal{B}$.

Since $v$ is subadditive, we have the following lower bound on the cumulative value of the subsets in $H$
\begin{align}\label{3}
\sum \limits _{i \in H} v(S_i) \geq \sum \limits _{j=1}^{|U|} v(S_j) - \frac{1}{3} {\rm F}(\mathcal{J}) \cdot |U| = \frac{2}{3}\ |U|  \cdot  {\rm F}(\mathcal{J})
\end{align}
\newtheorem{Claim}{Claim}
\begin{Claim}
Let $\mathcal{S} = (S_1,\ldots,S_{|U|})$ be the allocation computed by Feige's algorithm for input instance $\mathcal{J}.$ Then, every $S_i$, with the property that $v(S_i) \geq \frac{1}{3} {\rm F}(\mathcal{J})$, can be partitioned into at least $k=\left ( \frac{3v(S_i)}{{\rm F}(\mathcal{J})}-1\right )$ subsets $T^1_i,\ldots , T^k_i$ such that $v(T^j_i)\geq \frac{1}{20}{\rm F}(\mathcal{J})$ for each $T^j_i$.
\end{Claim}
\begin{proof}
Initialize $T_i^1$ to be the empty set. Then, we keep transferring goods---in any order and one at a time---from $S_i$ to $T_i^1$ till the value of $T_i^1$ goes over $\frac{1}{3}{\rm F}(\mathcal{J}).$ Returning the last such good back into $S_i$, the populated set $T_i^1$ satisfies 
\begin{align*}
v(T_i^1) & \geq \frac{1}{3}{\rm F}(\mathcal{J}) - \frac{1}{3.53}{\rm F}(\mathcal{J}) \tag{using (\ref{2}) and the subadditivity of $v$} \nonumber \\ & \geq \frac{1}{20}{\rm F}(\mathcal{J}) 
\end{align*}
Note that, by construction, $v(T_i^1) \leq \frac{1}{3}{\rm F}(\mathcal{J})$. Therefore, using the subadditivity of $v$, we get $v(S_i\setminus T_i^1) \geq v(S_i)-\frac{1}{3}{\rm F}(\mathcal{J})$.

We can repeat the above process to obtain subsets $T_i^1,\ldots T_i^k$ and stop when $v\left(S_i\setminus  ( T_i^1 \cup \ldots \cup T_i^k) \right) \leq \frac{1}{3}{\rm F}(\mathcal{J}).$ Given that, for each $T_i^j$, we remove a subset of value atmost $\frac{1}{3} {\rm F}(\mathcal{I})$ from $S_i$, the subadditivity of $v$ gives us $v \left(S_i\setminus ( T_i^1 \cup \ldots \cup T_i^k) \right) \geq v(S_i) - \frac{k{\rm F}(\mathcal{J})}{3}$. Hence, the following lower bound holds $k\geq \frac{3v(S_i)}{{\rm F}(\mathcal{J})} -1$.
In other words, we can extract at least $\frac{3v(S_i)}{{\rm F}(\mathcal{J})} -1$ bundles, each of value no less than $\frac{1}{20}{\rm F}(\mathcal{J})$, from $S_i$.
\end{proof}
We now apply the same procedure to every subset in $H \coloneqq \left\{ i \in \{1, 2, \ldots, |U| \} \mid v(S_i) \geq \frac{1}{3} {\rm F}(\mathcal{J}) \right\}$ to obtain $k^{'} = \sum \limits _{i=1}^{|H|} \left ( \frac{3v(S_i)}{{\rm F}(\mathcal{J})} -1 \right ) $ bundles, each of value at least $\frac{1}{20}{\rm F}(\mathcal{J})$. Using this equation and inequality (\ref{3}), we get
\begin{align*}
k^{'} & \geq \frac{2|U|{\rm F}(\mathcal{J})}{{\rm F}(\mathcal{J})} - |H| \geq 2|U| - |U|  \tag{Since $|H|<|U|$} \\
& = |U| \hspace{3.3cm}
\end{align*}

In conclusion, one can construct at least $|U|$ bundles of value at least $\frac{1}{20}{\rm F}(\mathcal{J}) \geq \frac{1}{40}{\rm M}_1(\mathcal{A}^*(\mathcal{J},1))$. Note that this observation implies that \textsc{AlgLow} successfully finds $|U|$ bundles each of value at least $\frac{1}{40} {\rm M}_1(\mathcal{A}^*(\mathcal{J},1))$. 

Proposition~\ref{Monotonicity} gives us ${\rm M}_1(\mathcal{A}^*(\mathcal{J},1)) \geq {\rm M}_p(\mathcal{A}^{*}(\mathcal{J},p))$ and, hence, the stated claim follows. 
\end{proof}

With an approximation guarantee for \textsc{AlgLow} in hand, we next analyze the first phase of the algorithm--specifically, analyze the while-loop in \textsc{Alg}. Note that in each iteration of this while-loop (Steps (\ref{Threshold}) to (\ref{step:end-while})) a good of value at least constant times the optimal $1$-mean welfare (of the current instance) is allocated as a singleton to an agent. We will show that these singleton assignments complement the subsequent use of  \textsc{AlgLow}, in the sense that the two phases together retain welfare guarantees. To formally prove Theorem \ref{MainTheorem}, we first state and prove two useful results, Lemma~\ref{Good_Transfer} (in Section~\ref{Supporting_Lemmas}) and Lemma~\ref{Induction_argument} (in Section~\ref{section:stitching-lemma}).

\section{Structural Lemma} \label{Supporting_Lemmas}
The following lemma provides a structural property of $p$-optimal allocations $\mathcal{A}^*(\mathcal{I}, p)$, for $p \in (-\infty,0.4)$. It states that the only way allocation $\mathcal{A}^*(\mathcal{I}, p)$ has a bundle $A^*_i(\mathcal{I}, p)$ of notably high value is through a single good $g \in A^*_i(\mathcal{I}, p)$ that by itself has high value. 


\begin{Lemma}\label{Good_Transfer}
Let $\mathcal{L}= \I{[M]}{[N]}$ be a fair-division instance wherein all the $N \in \mathbb{N}$ agents have an identical, subadditive valuation $v$ over the set of $M\in\mathbb{N}$ goods. In addition, let $\mathcal{A}^{*}(\mathcal{L},p) = \{ A^{*}_{1}(\mathcal{L},p), \cdots , A^{*}_{N}(\mathcal{L},p)\}$ be a $p$-mean optimal allocation in $\mathcal{L}$,  for any  $p \in \left( -\infty ,0.4 \right)$. 

If for any bundle $A^{*}_i(\mathcal{L},p)$, with $i \in [N]$, we have $v(A^{*}_i(\mathcal{L},p)) > 11.33\ {\rm F}(\mathcal{L})$, then there exists a good $g\in A_i^*(\mathcal{L},p)$ with the property that that $v(g)\geq \frac{1}{40} v(A_i^*(\mathcal{L},p)).$
\end{Lemma}

The proof of the above lemma is divided into  three parts (Sections~\ref{subsection:p-infty-zero-good-transfer}, \ref{subsection:p-zero-half-good-transfer}, and~\ref{subsection:p-zero-good-transfer}) depending on the range of the exponent parameter $p.$ 
\subsection{Proof of Lemma \ref{Good_Transfer} for $p \in (-\infty,0)$}
\label{subsection:p-infty-zero-good-transfer}

Assume, towards a contradiction, that $v(A^{*}_i(\mathcal{L},p)) > 11.33\ {\rm F}(\mathcal{L}) $, for some $i \in [N] $, and $v(g) \leq  \frac{1}{40} v(A^{*}_i(\mathcal{L},p))$ for all $g \in A^{*}_i(\mathcal{L},p).$ Recall that the $1$-mean welfare of the allocation returned by Feige's algorithm satisfies ${\rm F}(\mathcal{L}) \geq \frac{1}{2}{\rm M}_1({\mathcal{A}^{*}}(\mathcal{L},1))$. Pick a bundle $A^{*}_j(\mathcal{L},p)$ (in $\mathcal{A}^*(\mathcal{L}, p)$) with the property that $v(A^{*}_j(\mathcal{L},p)) \leq 2{\rm F}(\mathcal{L})$. Such a bundle exists, since  $2{\rm F}(\mathcal{L})\geq {\rm M}_1({\mathcal{A}^{*}}(\mathcal{L},1)) \geq {\rm M}_1(\mathcal{A}^*(\mathcal{L},p))$. 

Define a partition---$A'_i$ and $A'_j$---of $A^*_i$ as follows: \\
\noindent
(i) Initialize $A'_j$ to be the empty set. Then, we keep transferring goods from $A_i^*(\mathcal{L},p)$ to $A'_j$ (one at a time and in an arbitrary order) and stop as soon as the value of $A'_j$ exceeds $\frac{v(A^*_i(\mathcal{L},p))}{2}$.\\
(ii) Denote the remaining set of goods as $A'_i \coloneqq A^*_i(\mathcal{L},p) \setminus A'_j$. Note that, by construction, $v(A'_j)\geq \frac{v(A^{*}_i(\mathcal{L},p))}{2}$ and 
$v(A'_i)\geq \left (\frac{1}{2}-\frac{1}{40} \right ) v(A^*_i(\mathcal{L},p))$. The last inequality follows from the fact that $v$ is subadditive and the assumption that goods in $A^*_i(\mathcal{L},p)$ are of value at most $\frac{1}{40}v(A^*_i(\mathcal{L},p))$. \\
(iii) Write $\mathcal{B}= \{B_1, \ldots, B_N\}$ to denote the allocation obtained by replacing the bundles $A^{*}_i(\mathcal{L},p)$ and $A^{*}_j(\mathcal{L},p)$ in $\mathcal{A}^*(\mathcal{L},p)$ by $A'_i$ and $A^{*}_j(\mathcal{L},p) \cup A'_j$, respectively: $B_i = A'_i$ and $B_j = A^*_j(\mathcal{L},p) \cup A'_j$ along with $B_\ell = A^*_\ell (\mathcal{L},p)$ for all $\ell \in [N] \setminus \{i, j \}$. 

We will show that $\mathcal{B}$ has $p$-mean welfare strictly greater than that of the $p$-optimal allocation ${\mathcal{A}^{*}}(\mathcal{L},p)$. Hence, by way of contradiction, the desired result follows.

Recall that the current case addresses exponent parameters that are negative, $p \in (-\infty, 0)$. Hence, the following  inequality implies that the $p$-mean welfare of $\mathcal{B}$ is strictly greater than that of $\mathcal{A}^*(\mathcal{L},p)$:
\begin{align}
v(A'_i)^p+v(A'_j)^p< v(A^{*}_i(\mathcal{L},p))^p+v(A^{*}_j(\mathcal{L},p))^p \label{ineq:desired-neg-p}
\end{align}

However, for negative $p$, the lower bounds on the values of $A'_i$ and $A'_j$ gives us
\begin{align*}
v(A'_i)^p+v(A'_j)^p\leq \left (\frac{1}{2}-\frac{1}{40} \right )^p v(A^{*}_i(\mathcal{L},p))^p + \left (\frac{1}{2} \right )^pv(A^{*}_i(\mathcal{L},p))^p.
\end{align*}

In addition, using the bounds $v(A^*_j(\mathcal{L},p)) \leq 2 {\rm F}( \mathcal{L}) < \frac{2}{11.33} v(A^*_i(\mathcal{L},p))$, we get  
\begin{align*}
v(A^{*}_i(\mathcal{L},p))^p+ v(A^{*}_j(\mathcal{L},p))^p & > v(A^{*}_i(\mathcal{L},p))^p+ \left( \frac{2}{11.33} \right)^p v(A^{*}_i(\mathcal{L},p))^p.
\end{align*}

Therefore, the desired equation (\ref{ineq:desired-neg-p}) follows from the following numeric inequality, which is established in Appendix \ref{app 3}.
\begin{align*}
 \left (\frac{1}{2}-\frac{1}{40} \right )^p + \left (\frac{1}{2} \right )^p &\leq 1+ \left( \frac{2}{11.33} \right)^p \quad \text{ for all } p\in (-\infty,0).
\end{align*}
This establishes Lemma \ref{Good_Transfer} for $p \in (-\infty,0)$.

\subsection{Proof of Lemma \ref{Good_Transfer} for $p \in (0, 0.4)$}
\label{subsection:p-zero-half-good-transfer}
Assume, towards a contradiction, that $v(A^{*}_i(\mathcal{L},p)) > 11.33\ {\rm F}(\mathcal{L}) $, for some $i \in [N] $, and $v(g) \leq  \frac{1}{40} v(A^{*}_i(\mathcal{L},p))$ for all $g \in A^{*}_i(\mathcal{L},p).$ Recall that the $1$-mean welfare of the allocation returned by Feige's algorithm satisfies ${\rm F}(\mathcal{L}) \geq \frac{1}{2}{\rm M}_1({\mathcal{A}^{*}}(\mathcal{L},1))$. Pick a bundle $A^{*}_j(\mathcal{L},p)$ (in $\mathcal{A}^*(\mathcal{L}, p)$) with the property that $v(A^{*}_j(\mathcal{L},p)) \leq 2{\rm F}(\mathcal{L})$. Such a bundle exists, since  $2{\rm F}(\mathcal{L})\geq {\rm M}_1({\mathcal{A}^{*}}(\mathcal{L},1)) \geq {\rm M}_1(\mathcal{A}^*(\mathcal{L},p))$. 

Define a partition---$A'_i$ and $A'_j$---of $A^*_i(\mathcal{L},p)$ as follows: \\
\noindent
(i) Initialize $A'_j$ to be the empty set. Then, we keep transferring goods from $A_i^*(\mathcal{L},p)$ to $A'_j$ (one at a time and in an arbitrary order) and stop as soon as the value of $A'_j$ exceeds $\frac{v(A^*_i(\mathcal{L},p))}{2}$.\\
(ii) Denote the remaining set of goods as $A'_i \coloneqq A^*_i(\mathcal{L},p) \setminus A'_j$. Note that, by construction, $v(A'_j)\geq \frac{v(A^{*}_i(\mathcal{L},p))}{2}$ and 
$v(A'_i)\geq \left (\frac{1}{2}-\frac{1}{40} \right ) v(A^*_i(\mathcal{L},p))$. The last inequality follows from the fact that $v$ is subadditive and the assumption that goods in $A^*_i(\mathcal{L},p)$ are of value at most $\frac{1}{40}v(A^*_i(\mathcal{L},p))$. \\
(iii) Write $\mathcal{B}= \{B_1, \ldots, B_N\}$ to denote the allocation obtained by replacing the  bundles $A^{*}_i(\mathcal{L},p)$ and $A^{*}_j(\mathcal{L},p)$ in $\mathcal{A}^*(\mathcal{L},p)$ by $A'_i$ and $A^{*}_j(\mathcal{L},p) \cup A'_j$, respectively: $B_i = A'_i$ and $B_j = A^*_j(\mathcal{L},p) \cup A'_j$ along with $B_\ell = A^*_\ell (\mathcal{L},p)$ for all $\ell \in [N] \setminus \{i, j \}$. 

We will show that $\mathcal{B}$ has $p$-mean welfare strictly greater than that of the $p$-optimal allocation ${\mathcal{A}^{*}}(\mathcal{L},p)$. Hence, by way of contradiction, the desired result follows.

Notice that the construction of $A'_i$, $A'_j$ and $\mathcal{B}$ hold for $p=0$ as well. We will use these sets in the next section to prove an analogous result for Nash social welfare. 

The current case addresses exponent parameters that are positive $p \in (0, 0.4)$. Hence, the following  inequality implies that the $p$-mean welfare of $\mathcal{B}$ is strictly greater than that of $\mathcal{A}^*(\mathcal{L},p)$
\begin{align}
v(A'_i)^p+v(A'_j)^p & > v(A^{*}_i(\mathcal{L},p))^p+v(A^{*}_j(\mathcal{L},p))^p \label{ineq:desired-pos-p}
\end{align}

However, for positive $p$, the lower bounds on the values of $A'_i$ and $A'_j$ gives us
\begin{align*}
v(A'_i)^p+v(A'_j)^p & \geq \left (\frac{1}{2}-\frac{1}{40} \right )^p v(A^{*}_i(\mathcal{L},p))^p + \left (\frac{1}{2} \right )^pv(A^{*}_i(\mathcal{L},p))^p.
\end{align*}

In addition, using the bounds $v(A^*_j(\mathcal{L},p)) \leq 2 {\rm F}( \mathcal{L}) < \frac{2}{11.33} v(A^*_i(\mathcal{L},p))$, we get  
\begin{align*}
v(A^{*}_i(\mathcal{L},p))^p+ v(A^{*}_j(\mathcal{L},p))^p & < v(A^{*}_i(\mathcal{L},p))^p + \left( \frac{2}{11.33} \right)^pv(A^{*}_i(\mathcal{L},p))^p.
\end{align*}

Therefore, the desired equation (\ref{ineq:desired-pos-p}) follows from the following numeric inequality, which is established in Appendix \ref{app 3}.
\begin{align*}
 \left (\frac{1}{2}-\frac{1}{40} \right )^p + \left (\frac{1}{2} \right )^p &\geq 1+ \left( \frac{2}{11.33} \right)^p \hbox{ for } p\in (0,0.4).
\end{align*}

This establishes Lemma \ref{Good_Transfer} for $p \in (0, 0.4)$.


\subsection{Proof of Lemma \ref{Good_Transfer} for Nash Social Welfare ($p=0$)}
\label{subsection:p-zero-good-transfer}

Recall the sets $A'_i$, $A'_j$, and the allocation $\mathcal{B}$ defined in Section \ref{subsection:p-zero-half-good-transfer}. We will show that $\mathcal{B}$ has Nash welfare strictly greater than that of ${\mathcal{A}^{*}}(\mathcal{L}, 0)$, and hence, by way of contradiction, establish the desired result.

 In order to obtain ${\rm M}_0(\mathcal{B})>{\rm M}_0(\mathcal{A}^*(\mathcal{L},0))$, it suffices to prove that $v(A'_i)v(A'_j)>v\left( A^{*}_i(\mathcal{L},0)\right)v( A^{*}_j(\mathcal{L},0)).$
The lower bounds we obtained on the values of $A'_i$ and $A'_j$ gives us 
\begin{align*}
v(A'_i)v(A'_j)&\geq \left (\frac{1}{2}-\frac{1}{40} \right )v\left( A^{*}_i (\mathcal{L},0)\right) \frac{v\left( A^{*}_i (\mathcal{L},0)\right)}{2}
\geq  0.2 \ v\left( A^{*}_i (\mathcal{L},0)\right)^2 
\end{align*}

In addition, we have 
\begin{align*}
v\left( A^{*}_i (\mathcal{L},0)\right) v\left( A^{*}_j (\mathcal{L},0)\right)  & <\frac{2}{11.33} \  v\left( A^{*}_i (\mathcal{L},0)\right)^2
<0.18 \ v\left( A^{*}_i (\mathcal{L},0)\right)^2 
\end{align*}
Therefore, the lemma holds for $p=0$ as well. 


\section{Combination Lemma} \label{section:stitching-lemma}
The following lemma shows that the goods assigned as singletons in the while-loop of $\textsc{Alg}$ (Algorithm~\ref{Alg}), along with a $p$-optimal allocation of instance $\mathcal{J}$ that remains at the termination of the loop, lead to a $p$-mean welfare that is comparable to the optimal, for all $p \in (-\infty, 0.4)$.

As shown previously in Lemma~\ref{Low_valued}, $\textsc{AlgLow}$---with instance $\mathcal{J}$ as input---achieves a constant-factor approximation for the $p$-mean welfare objective. Hence, the lemma established in this section will enable us to combine the welfare guarantees of the goods assigned in the while-loop of \textsc{Alg} and the allocation computed by \textsc{AlgLow} to obtain the desired approximation result for $p \in (-\infty, 0.4)$. 

Specifically, given a fair-division instance $\mathcal{I} = \langle [m], [n], v \rangle$ as input, let $\{g_1, \ldots, g_k\}$ denote the set of goods that get assigned as singletons in the while-loop of $\textsc{Alg}$ (Algorithm~\ref{Alg}). Furthermore, for $1 \leq t \leq k$, let $\mathcal{I}^t$ denote the instance obtained at the end of the $t$th iteration of this while-loop. Since in the first $t$ iterations $\textsc{Alg}$ assigns goods $\{g_1, \ldots, g_t\}$ to the first $t$ agents as singletons, we have $\mathcal{I}^t = \langle [m]\setminus \{g_1, \ldots, g_t \}, [n]\setminus \{1, \ldots, t\}, v \rangle$. In particular, $\mathcal{J} \coloneqq \mathcal{I}^k$ is the instance that remains after the termination of the while-loop in $\textsc{Alg}$ and this instance is passed on to $\textsc{AlgLow}$ as input.

Instance $\mathcal{J}$ consists of $n-k$ agents and, hence, in $\mathcal{J}$, any $p$-optimal allocation $\mathcal{A}^*(\mathcal{J},p)$ contains $(n-k)$ bundles. For notational convenience, we will index these bundles from $k+1$ to $n$, i.e., $\mathcal{A}^*(\mathcal{J},p)= \left( A^*_{k+1} (\mathcal{J},p), \ldots, A^*_{n} (\mathcal{J},p) \right)$.


\begin{Lemma}\label{Induction_argument}
Given a fair-division instance $\mathcal{I} = \langle [m], [n], v \rangle$ with an identical subadditive valuation $v$, let $\{g_1, \ldots, g_k\}$ denote the set of goods that get assigned as singletons in the while-loop of $\textsc{Alg}$ and let $\mathcal{J}=\langle [m] \setminus \{g_1, \ldots, g_k \}, [n] \setminus [k], v \rangle$ be the instance that remains after the termination of this loop. In addition, let $\mathcal{A}^*(\mathcal{I},p)=\left(A^*_1 (\mathcal{I},p), \ldots, A^*_n(\mathcal{I},p) \right)$ and $\mathcal{A}^*(\mathcal{J},p) = \left( A^*_{k+1} (\mathcal{J},p), \ldots, A^*_{n} (\mathcal{J},p) \right)$ denote $p$-optimal allocations of instances $\mathcal{I}$ and $\mathcal{J}$, respectively. Then, with constant $\alpha = 40$,   
\begin{itemize}
\item For $p \in (-\infty, 0)$, we have $\alpha^p  \sum \limits_{i=1}^k v(g_i)^p \ + \ \sum \limits_{j=k+1}^{n} v (A^{*}_j(\mathcal{J},p) )^p \leq \sum \limits_{i=1}^n v(A^{*}_i(\mathcal{I},p ))^p$.  

\item For $p \in (0,0.4)$, we have $\alpha^p \sum \limits_{i=1}^k v \left( g_i\right)^p  \ +  \ \sum \limits_{j=k+1}^{n} v\left(A^{*}_j(\mathcal{J},p)\right)^p \geq \sum \limits_{i=1}^n v(A^{*}_i(\mathcal{I},p ))^p$.
\item  For $p=0$, we have $\alpha^k \prod \limits_{i=1}^k v(g_i) \  \prod \limits_{j=k+1}^{n} v(A^*_j (\mathcal{J},p) ) \geq  \prod \limits_{i=1}^{n} v(A^*_i(\mathcal{I},p))$. 
 
\end{itemize}
\end{Lemma}
We will prove Lemma~\ref{Induction_argument} by considering different ranges of the exponent parameter $p$ separately. However, in all of the ranges, the desired inequality is obtained by inducting on the number of iterations of the while-loop in \textsc{Alg}.

\subsection{Proof of Lemma \ref{Induction_argument} for $p\in(-\infty,0)$}
For $0 \leq t \leq k$, recall that instance $\mathcal{I}^t \coloneqq \langle [m]\setminus \{g_1, \ldots, g_t \}, [n]\setminus \{1, \ldots, t\}, v \rangle$ and its corresponding $p$-mean optimal, $\mathcal{A}^*(\mathcal{I}^t,p)=\left( A^*_{t+1}(\mathcal{I}^{t},p),\ldots , A^*_n(\mathcal{I}^t,p)\right)$. 

We prove by induction over all $0\leq t \leq k, $ that
\begin{align}\label{7}
(40) ^p\sum \limits_{i=1}^t v \left( g_i\right)^p + \sum \limits_{j=t+1}^{n} v\left(A^{*}_j(\mathcal{I}^t,p)\right)^p & \leq \sum \limits_{i=1}^n v(A^{*}_i(\mathcal{I},p))^p
\end{align}
\emph{Base Case:} When $t=0,$ we have $\mathcal{I}^t = \mathcal{I}$ and, hence, both sides of equation (\ref{7}) are equal to each other. Therefore, the base case holds.

\noindent
\emph{Induction Step:} We establish inequality  (\ref{7}) for $t$, assuming that it holds for $t-1$. 

 Consider the good $g_{t}$ that was assigned in the $t${th} iteration of the while-loop in \textsc{Alg}. Note that $v(g_{t}) \geq \frac{1}{3.53}{\rm F}(\mathcal{I}^{t-1})$ (see Step \ref{Threshold}). Without loss of generality, we may assume that $g_{t} \in A^{*}_{t}(\mathcal{I}^{t-1},p)$. This assumption is justified since ${\mathcal{A}^{*}}(\mathcal{I}^{t-1},p)$ is a $p$-optimal allocation of the instance $\mathcal{I}^{t-1}$ and, hence, ${\mathcal{A}^{*}}(\mathcal{I}^{t-1},p)$ is an $(n-t+1)$-partition of the goods $[m] \setminus \{g_1, g_2, \ldots, g_{t-1}\}$, i.e., $g_t$ belongs to one of bundles in ${\mathcal{A}^{*}}(\mathcal{I}^{t-1},p) =\left( A^*_{t}(\mathcal{I}^{t-1},p),\ldots , A^*_n(\mathcal{I}^{t-1},p) \right)$. 
 
 We lower bound the value of $g_t$ in terms of the value of the bundle $A^*_t(\mathcal{I}^{t-1},p)$.

\noindent {Case {\rm I}:} $v(A^*_{t} (\mathcal{I}^{t-1},p)) \leq 11.33 \ {\rm F}(\mathcal{I}^{t-1})$. In this case, we have $v(g_{t}) \geq \frac{1}{3.53} {\rm F}(\mathcal{I}^{t-1}) \geq \frac{1}{40} v({A_{t}^{*}}(\mathcal{I}^{t-1},p))$; since, $3.53 \times 11.33 \leq 40$. \\

\noindent {Case {\rm II}:} $v(A^*_{t}(\mathcal{I}^{t-1},p)) > 11.33 \ {\rm F}(\mathcal{I}^{t-1})$. Recall that the goods are indexed in non-increasing order of value (see Step \ref{Ordered_Goods} of \textsc{Alg}) and, hence, $g_t$ is the highest valued good in the instance $\mathcal{I}^{t-1}$. Therefore, Lemma \ref{Good_Transfer} gives us 
\begin{align}
v(g_{t})\geq \frac{1}{40} v(A^*_{t}(\mathcal{I}^{t-1},p )) \label{8}
\end{align}
Note that inequality (\ref{8}) holds in both Cases {\rm I} and {\rm II} mentioned above. Furthermore, since the current case addresses negative $p \in (-\infty, 0)$, we have $ (40)^p \ v(g_{t})^p\leq v(A^*_{t}(\mathcal{I}^{t-1},p  ))^p$. 

We add $(40) ^p \sum \limits_{i=1}^{t-1} v \left( g_i\right)^p + \sum \limits_{j=t+1}^{n} v\left(A^{*}_j(\mathcal{I}^{t},p)\right)^p $ to both sides of the previous inequality to obtain

\begin{align}
(40) ^p \sum \limits_{i=1}^{t} v \left( g_i\right)^p + \sum \limits_{j=t+1}^{n} v\left(A^{*}_j(\mathcal{I}^{t},p)\right)^p & \leq (40) ^p \sum \limits_{i=1}^{t-1} v \left( g_i\right)^p +v(A^*_{t}(\mathcal{I}^{t-1},p  ))^p +\sum \limits_{j=t+1}^{n} v\left(A^{*}_j(\mathcal{I}^{t},p)\right)^p \label{ineq:interim}
\end{align}

Note that the allocation ${\mathcal{A}^{*}}(\mathcal{I}^{t-1},p)$ is defined over the goods $[m]\setminus \{g_1,\ldots , g_{t-1}\}$ and the bundle $A^{*}_{t}(\mathcal{I}^{t-1},p )$ contains $g_{t}$. On the other hand, ${\mathcal{A}^{*}}(\mathcal{I}^{t},p)$ is defined over $[m]\setminus \{g_1,\ldots , g_{t}\}$. Hence, all the goods in ${\mathcal{A}^*}(\mathcal{I}^{t-1},p)$, with the exception of $g_{t}$, appear in ${\mathcal{A}^{*}}(\mathcal{I}^{t},p)$. 

In other words, the last $n-t$ bundles of $\mathcal{A}^*(\mathcal{I}^{t-1},p)$ and all the $n-t$ bundles of $\mathcal{A}^*(\mathcal{I}^{t},p)$ satisfy $\bigcup \limits _{j=t+1}^n A^{*}_j(\mathcal{I}^{t-1},p) \subseteq \bigcup \limits _{j=t+1}^n A^{*}_j(\mathcal{I}^{t},p)$. Using this containment and the fact that ${\mathcal{A}^{*}}(\mathcal{I}^{t},p)$ is the $p$-optimal allocation for the instance $\mathcal{I}^{t}$, we have
\begin{align*}
\left(\frac{1}{n-t} \sum \limits_{j=t+1}^{n} v\left(A^{*}_j(\mathcal{I}^{t-1},p)\right)^p\right)^\frac{1}{p} & \leq \left( \frac{1}{n-t} \sum \limits_{j=t+1}^{n} v\left(A^{*}_j( \mathcal{I}^{t},p)\right)^p \right)^\frac{1}{p}.
\end{align*}

Exponentiating both sides by $p$ (which in the current case is negative) and multiplying by $n-t$, gives us $ \sum \limits_{j=t+1}^{n} v\left(A^{*}_j(\mathcal{I}^{t-1},p)\right)^p \geq \sum \limits_{j=t+1}^{n} v\left(A^{*}_j(\mathcal{I}^{t},p)\right)^p .$ Therefore, inequality (\ref{ineq:interim}) extends to

\begin{align*}
(40) ^p \sum \limits_{i=1}^{t} v \left( g_i\right)^p + \sum \limits_{j=t+1}^{n} v\left(A^{*}_j(\mathcal{I}^{t},p)\right)^p & \leq 40 ^p \sum \limits_{i=1}^{t-1} v \left( g_i\right)^p + \sum \limits_{j=t}^{n} v\left(A^{*}_j(\mathcal{I}^{t-1},p)\right)^p \\
& \leq \sum \limits _{i=1}^n v(A^{*}_i(\mathcal{I},p))^p
\end{align*}
The last inequality follows from the induction hypothesis. Setting $t=k$ gives us the desired inequality for $p\in(-\infty,0)$
\begin{align}\label{9}
( 40 )^p\sum \limits_{i=1}^k v \left( g_i\right)^p + \sum \limits_{j=k+1}^{n} v\left(A^{*}_j(\mathcal{I}^k,p)\right)^p & \leq \sum \limits_{i=1}^n v(A^{*}_i(\mathcal{I},p))^p 
\end{align}
Therefore, Lemma \ref{Induction_argument} holds for $p\in (-\infty,0)$.


\subsection{Proof of Lemma \ref{Induction_argument} for $p\in (0,0.4)$}
For $0 \leq t \leq k$, recall that instance $\mathcal{I}^t \coloneqq \langle [m]\setminus \{g_1, \ldots, g_t \}, [n]\setminus \{1, \ldots, t\}, v \rangle$ and its corresponding $p$-mean optimal, $\mathcal{A}^*(\mathcal{I}^t,p)=\left( A^*_{t+1}(\mathcal{I}^{t},p),\ldots , A^*_n(\mathcal{I}^t,p)\right)$. 

We prove by induction over all $0\leq t \leq k, $ that
\begin{align} 
(40) ^p\sum \limits_{i=1}^t v \left( g_i\right)^p + \sum \limits_{j=t+1}^{n} v\left(A^{*}_j(\mathcal{I}^t,p)\right)^p & \geq \sum \limits_{i=1}^n v(A^{*}_i(\mathcal{I},p))^p \label{ineq:ind}
\end{align}
\emph{Base Case:} When $t=0,$ we have $\mathcal{I}^t = \mathcal{I}$ and, hence, both sides of equation (\ref{7}) are equal to each other. Therefore, the base case holds.

\noindent
\emph{Induction Step:} We establish inequality  (\ref{ineq:ind}) for $t$, assuming that it holds for $t-1$. 

Consider the good $g_{t}$ that was assigned in the $t${th} iteration of the while-loop in \textsc{Alg}. Note that $v(g_{t}) \geq \frac{1}{3.53}{\rm F}(\mathcal{I}^{t-1})$ (see Step \ref{Threshold}). Without loss of generality, we may assume that $g_{t} \in A^{*}_{t}(\mathcal{I}^{t-1},p)$. This assumption is justified since ${\mathcal{A}^{*}}(\mathcal{I}^{t-1},p)$ is a $p$-optimal allocation of the instance $\mathcal{I}^{t-1}$ and, hence, ${\mathcal{A}^{*}}(\mathcal{I}^{t-1},p)$ is an $(n-t+1)$-partition of the goods $[m] \setminus \{g_1, g_2, \ldots, g_{t-1}\}$, i.e., $g_t$ belongs to one of bundles in ${\mathcal{A}^{*}}(\mathcal{I}^{t-1},p) =\left( A^*_{t}(\mathcal{I}^{t-1},p),\ldots , A^*_n(\mathcal{I}^{t-1},p) \right)$. \\

We lower bound the value of $g_t$ in terms of the value of the bundle $A^*_t(\mathcal{I}^{t-1},p)$.

\noindent {Case {\rm I}:} $v(A^*_{t} (\mathcal{I}^{t-1},p)) \leq 11.33 \ {\rm F}(\mathcal{I}^{t-1})$. In this case, we have $v(g_{t}) \geq \frac{1}{3.53} {\rm F}(\mathcal{I}^{t-1}) \geq \frac{1}{40} v({A_{t}^{*}}(\mathcal{I}^{t-1},p))$; since, $3.53 \times 11.33 \leq 40$.  

\noindent {Case {\rm II}:} $v(A^*_{t}(\mathcal{I}^{t-1},p)) > 11.33 \ {\rm F}(\mathcal{I}^{t-1})$. Recall that the goods are indexed in non-increasing order of value (see Step \ref{Ordered_Goods} of \textsc{Alg}) and, hence, $g_t$ is the highest valued good in the instance $\mathcal{I}^{t-1}$. Therefore, Lemma \ref{Good_Transfer} gives us 
\begin{align}
v(g_{t})\geq \frac{1}{40} v(A^*_{t}(\mathcal{I}^{t-1},p )) \label{ineq:g-t-val}
\end{align}

Note that inequality (\ref{ineq:g-t-val}) holds in both Cases {\rm I} and {\rm II} mentioned above. Furthermore, since the current case addresses positive $p\in (0,0.4)$, we have $ (40)^p \ v(g_{t})^p \geq v(A^*_{t}(\mathcal{I}^{t-1},p  ))^p$. 

We add $(40) ^p \sum \limits_{i=1}^{t-1} v \left( g_i\right)^p + \sum \limits_{j=t+1}^{n} v\left(A^{*}_j(\mathcal{I}^{t},p)\right)^p $ to both sides of the previous inequality to obtain

\begin{align}
(40) ^p \sum \limits_{i=1}^{t} v \left( g_i\right)^p + \sum \limits_{j=t+1}^{n} v\left(A^{*}_j(\mathcal{I}^{t},p)\right)^p & \geq (40) ^p \sum \limits_{i=1}^{t-1} v \left( g_i\right)^p +v(A^*_{t}(\mathcal{I}^{t-1},p  ))^p +\sum \limits_{j=t+1}^{n} v\left(A^{*}_j(\mathcal{I}^{t},p)\right)^p \label{ineq:interim-positive}
\end{align}

Note that the allocation ${\mathcal{A}^{*}}(\mathcal{I}^{t-1},p)$ is defined over the goods $[m]\setminus \{g_1,\ldots , g_{t-1}\}$ and the bundle $A^{*}_{t}(\mathcal{I}^{t-1},p )$ contains $g_{t}$. On the other hand, ${\mathcal{A}^{*}}(\mathcal{I}^{t},p)$ is defined over $[m]\setminus \{g_1,\ldots , g_{t}\}$. Hence, all the goods in ${\mathcal{A}^*}(\mathcal{I}^{t-1},p)$, with the exception of $g_{t}$, appear in ${\mathcal{A}^{*}}(\mathcal{I}^{t},p)$. 

In other words, the last $n-t$ bundles of $\mathcal{A}^*(\mathcal{I}^{t-1},p)$ and all the $n-t$ bundles of $\mathcal{A}^*(\mathcal{I}^{t},p)$ satisfy $\bigcup \limits _{j=t+1}^n A^{*}_j(\mathcal{I}^{t-1},p) \subseteq \bigcup \limits _{j=t+1}^n A^{*}_j(\mathcal{I}^{t},p)$. Using this containment and the fact that ${\mathcal{A}^{*}}(\mathcal{I}^{t},p)$ is the $p$-optimal allocation for the instance $\mathcal{I}^{t}$, we have
\begin{align*}
\left(\frac{1}{n-t} \sum \limits_{j=t+1}^{n} v\left(A^{*}_j(\mathcal{I}^{t-1},p)\right)^p\right)^\frac{1}{p} & \leq \left( \frac{1}{n-t} \sum \limits_{j=t+1}^{n} v\left(A^{*}_j( \mathcal{I}^{t},p)\right)^p \right)^\frac{1}{p}.
\end{align*}

Exponentiating both sides by $p$ (which in the current case is positive) and multiplying by $n-t$, gives us $ \sum \limits_{j=t+1}^{n} v\left(A^{*}_j(\mathcal{I}^{t-1},p)\right)^p \leq \sum \limits_{j=t+1}^{n} v\left(A^{*}_j(\mathcal{I}^{t},p)\right)^p .$ Therefore, via inequality (\ref{ineq:interim-positive}), we obtain
\begin{align*}
(40) ^p \sum \limits_{i=1}^{t} v \left( g_i\right)^p + \sum \limits_{j=t+1}^{n} v\left(A^{*}_j(\mathcal{I}^{t},p)\right)^p & \geq 40 ^p \sum \limits_{i=1}^{t-1} v \left( g_i\right)^p + \sum \limits_{j=t}^{n} v\left(A^{*}_j(\mathcal{I}^{t-1},p)\right)^p \\
& \geq \sum \limits _{i=1}^n v(A^{*}_i(\mathcal{I},p))^p
\end{align*}
The last inequality follows from the induction hypothesis. Setting $t=k$ gives us the desired inequality for $p \in (0, 0.4)$
\begin{align} 
( 40 )^p\sum \limits_{i=1}^k v \left( g_i\right)^p + \sum \limits_{j=k+1}^{n} v\left(A^{*}_j(\mathcal{I}^k,p)\right)^p & \geq \sum \limits_{i=1}^n v(A^{*}_i(\mathcal{I},p))^p 
\end{align}
Therefore, Lemma \ref{Induction_argument} holds for $p \in (0, 0.4)$. 

\subsection{Proof of Lemma \ref{Induction_argument} for Nash Social Welfare ($p=0$)}
\label{subsection:p-zero}

For $0 \leq t \leq k$, recall that instance $\mathcal{I}^t \coloneqq \langle [m]\setminus \{g_1, \ldots, g_t \}, [n]\setminus \{1, \ldots, t\}, v \rangle$ and its corresponding $0$-mean optimal (i.e., Nash optimal), $\mathcal{A}^*(\mathcal{I}^t, 0)=\left( A^*_{t+1}(\mathcal{I}^{t},0),\ldots , A^*_n(\mathcal{I}^t,0)\right)$. 

We prove by induction over all $0\leq t \leq k, $ that

\begin{align}\label{13}
 \prod \limits_{i=1}^t v(g_i) \  \prod \limits_{j=t+1}^{n} v(A^*_j (\mathcal{I}^t,0) ) \geq  \left( \frac{1}{40}\right)^t \prod \limits_{i=1}^{n} v(A^*_i(\mathcal{I},0)).
\end{align}

\noindent
\emph{Base Case:} When $t=0,$ we have $\mathcal{I}^t = \mathcal{I}$ and, hence, both sides of equation (\ref{13}) are equal to each other. Therefore, the base case holds. \\

\noindent
\emph{Induction Step:} We establish inequality  (\ref{13}) for $t$, assuming that it holds for $t-1$.  

 Consider the good $g_{t}$ that was assigned in the $t${th} iteration of the while-loop in \textsc{Alg}. Note that $v(g_{t}) \geq \frac{1}{3.53}{\rm F}(\mathcal{I}^{t-1})$ (see Step \ref{Threshold}). Without loss of generality, we may assume that $g_{t} \in A^{*}_{t}(\mathcal{I}^{t-1},0)$. This assumption is justified since ${\mathcal{A}^{*}}(\mathcal{I}^{t-1}, 0)$ is a $0$-optimal (Nash optimal) allocation of the instance $\mathcal{I}^{t-1}$ and, hence, ${\mathcal{A}^{*}}(\mathcal{I}^{t-1}, 0)$ is an $(n-t+1)$-partition of the goods $[m] \setminus \{g_1, g_2, \ldots, g_{t-1}\}$, i.e., $g_t$ belongs to one of bundles in ${\mathcal{A}^{*}}(\mathcal{I}^{t-1}, 0) =\left( A^*_{t}(\mathcal{I}^{t-1},0),\ldots , A^*_n(\mathcal{I}^{t-1},0) \right)$. \\
 
 We lower bound the value of $g_t$ in terms of the value of the bundle $A^*_t(\mathcal{I}^{t-1},0)$.

\noindent {Case {\rm I}:} $v(A^*_{t} (\mathcal{I}^{t-1},0)) \leq 11.33 \ {\rm F}(\mathcal{I}^{t-1})$. In this case, we have $v(g_{t}) \geq \frac{1}{3.53} {\rm F}(\mathcal{I}^{t-1}) \geq \frac{1}{40} v({A_{t}^{*}}(\mathcal{I}^{t-1},0))$; since, $3.53 \times 11.33 \leq 40$. 

\noindent {Case {\rm II}:} $v(A^*_{t}(\mathcal{I}^{t-1},0)) > 11.33 \ {\rm F}(\mathcal{I}^{t-1})$. Recall that the goods are indexed in non-increasing order of value (see Step \ref{Ordered_Goods} of \textsc{Alg}) and, hence, $g_t$ is the largest valued good in the instance $\mathcal{I}^{t-1}$. Therefore, Lemma \ref{Good_Transfer} gives us 
\begin{align}
v(g_{t})\geq \frac{1}{40} v(A^*_{t}(\mathcal{I}^{t-1},0 )) \label{ineq:val-g-t-0}
\end{align}

Here, inequality (\ref{ineq:val-g-t-0}) holds in both Cases {\rm I} and {\rm II} mentioned above. 

Note that the allocation ${\mathcal{A}^{*}}(\mathcal{I}^{t-1},0)$ is defined over the goods $[m]\setminus \{g_1,\ldots , g_{t-1}\}$ and the bundle $A^{*}_{t}(\mathcal{I}^{t-1},0 )$ contains $g_{t}$. On the other hand, ${\mathcal{A}^{*}}(\mathcal{I}^{t}, 0)$ is defined over $[m]\setminus \{g_1,\ldots , g_{t}\}$. Hence, all the goods in ${\mathcal{A}^*}(\mathcal{I}^{t-1}, 0)$, with the exception of $g_{t}$, appear in ${\mathcal{A}^{*}}(\mathcal{I}^{t}, 0)$. 

In other words, the last $n-t$ bundles of $\mathcal{A}^*(\mathcal{I}^{t-1}, 0)$ and all the $n-t$ bundles of $\mathcal{A}^*(\mathcal{I}^{t}, 0)$ satisfy $\bigcup \limits _{j=t+1}^n A^{*}_j(\mathcal{I}^{t-1},0) \subseteq \bigcup \limits _{j=t+1}^n A^{*}_j(\mathcal{I}^{t},0)$. Using this containment and the fact that ${\mathcal{A}^{*}}(\mathcal{I}^{t}, 0)$ is the $0$-optimal (Nash optimal) allocation for the instance $\mathcal{I}^{t}$, we have
\begin{align}
\prod\limits_{j=t+1}^n v\left(A^{*}_j(\mathcal{I}^{t},0)\right ) > \prod\limits_{j=t+1}^n v\left(A^{*}_j(\mathcal{I}^{t-1},0)\right ) \label{ineq:prod-induction}
\end{align}

Multiplying by $v(g_1) \ v(g_2) \ldots v(g_{t-1}) \ v(g_{t})$ on both sides of the previous inequality and using equation (\ref{ineq:val-g-t-0}), we get
\begin{align*}
v(g_1) \ v(g_2) \ldots v(g_{t-1}) \ v(g_{t}) \prod\limits_{j=t+1}^n v\left(A^{*}_j(\mathcal{I}^{t},0)\right) & \geq \frac{1}{40}v\left(g_1\right)\ldots v\left(g_{t-1}\right) \ v\left(A^{*}_{t}(\mathcal{I}^{t-1},0)\right) \prod\limits_{j=t+1}^n v\left(A^{*}_j(\mathcal{I}^{t},0)\right) \\
& \geq \frac{1}{40}v\left(g_1\right)\ldots v\left(g_{t-1}\right) \ v\left(A^{*}_{t}(\mathcal{I}^{t-1},0)\right) \prod\limits_{j=t+1}^n v\left(A^{*}_j(\mathcal{I}^{t-1},0)\right) \tag{via inequality (\ref{ineq:prod-induction})}\\
  &\geq \left(\frac{1}{40}\right)^{t} \prod \limits_{i=1}^{n} v(A^*_i(\mathcal{I},0))  \tag{using the induction hypothesis}
\end{align*}

Setting $t=k$ gives us the desired inequality for $p =0$
\begin{align*}
(40)^k \prod \limits_{i=1}^k v(g_i) \  \prod \limits_{j=k+1}^{n} v(A^*_j (\mathcal{I}^k,0) ) \geq  \prod \limits_{i=1}^{n} v(A^*_i(\mathcal{I},0)).
\end{align*}
Therefore, Lemma \ref{Induction_argument} holds for $p =0$ as well. \\

Using the lemmas established in the previous sections, we now prove Theorem \ref{MainTheorem} for the exponent parameter $p\in (-\infty , 0.4)$.

\section{Proof of Theorem \ref{MainTheorem} for $p \in (-\infty,0.4)$}
\label{subsection:p-inf-half}

As before, we will write $\mathcal{I}=\I{[m]}{[n]}$ to denote the given fair-division instance. Also, write $\mathcal{J} = \mathcal{I}^k = \I{[m]\setminus \{g_1,\ldots,g_k\}}{[n]\setminus [k]}$ to denote the instance obtained at the termination of the while loop in \textsc{Alg}. That is, instance $\mathcal{J}$ is obtained after allocating the $k$ highest-valued goods as singletons to different agents and $\mathcal{J}$ is passed on as an input to \textsc{AlgLow}. 

Also recall that $\mathcal{A}^*(\mathcal{I},p)=(A^*_1(\mathcal{I},p),\ldots, A^*_n(\mathcal{I},p))$ and $\mathcal{A}^*(\mathcal{J}, p)= \left( A^*_{k+1}(\mathcal{J},p),\ldots, A^*_n(\mathcal{J},p) \right)$ denote the $p$-mean optimal allocations of instances $\mathcal{I}$ and $\mathcal{J}$, respectively.

So far, we have established two results \\
\noindent
(i) Lemma~\ref{Induction_argument}: The allocation $(\{g_1\},\ldots ,\{g_k\}, A^*_{k+1}(\mathcal{J},p), \ldots, A^*_{n}(\mathcal{J},p) )$ achieves welfare comparable to the optimal $p$-mean welfare (i.e., comparable to ${\rm M}_p \left(\mathcal{A}^*(\mathcal{I},p) \right)$), for $p \in (-\infty, 0.4)$. \\
\noindent
(ii) Lemma~\ref{Low_valued}: For the instance $\mathcal{J}$ and any $p\in (-\infty, 0.4]$, \textsc{AlgLow} computes an allocation $\mathcal{B}= (B_{k+1},\ldots,B_{n})$ such that, for all $j \in \{k+1, \ldots, n\}$,\footnote{For notational convenience, we index the bundles in allocation $\mathcal{B}$ from $k+1$ to $n$.}
\begin{align}\label{16}
v(B_j)\geq\frac{1}{40}{\rm M}_1(\mathcal{A}^{*}(\mathcal{J},1))\geq\frac{1}{40}{\rm M}_p(\mathcal{A}^{*}(\mathcal{J},p)) 
\end{align}

The allocation returned by \textsc{Alg} for input instance $\mathcal{I}$ is $\mathcal{A}=(\{g_1\},\ldots,\{g_k\}, B_{k+1}, \ldots ,B_n)$. We will prove Theorem~\ref{MainTheorem}, for $p \in (-\infty,0.4)$, by showing that Lemma~\ref{Induction_argument} and Lemma~\ref{Low_valued}, in conjunction, imply that the $p$-mean welfare of $\mathcal{A}$ is a constant times the optimal $p$-mean optimal; specifically, ${\rm M}_p(\mathcal{A})\geq \frac{1}{40}{\rm M}_p(\mathcal{A}^{*}(\mathcal{I},p))$, for $p\in(-\infty,0.4)$.

We split the proof of this inequality into three parts, depending on the range of the exponent parameter $p$.

\noindent
\textbf{Case 1:} $p \in (-\infty, 0)$. Since in this case $p$ is negative, to obtain the desired inequality, ${\rm M}_p(\mathcal{A})\geq \frac{1}{40} {\rm M}_p(\mathcal{A}^{*}(\mathcal{I},p))$, it suffices to show that
\begin{align*}
\sum \limits_{i=1}^k v \left( g_i\right)^p + \sum \limits_{j=k+1}^{n} v\left(B_j\right)^p\leq \frac{1}{(40) ^p}\sum \limits_{i=1}^n v(A^{*}_i(\mathcal{I},p))^p 
\end{align*}

Exponentiating both sides of equation (\ref{16}) by $p <0$ and summing over $j \in \{k+1, \ldots, n\}$ lead to 
\begin{align}
\sum \limits _{j=k+1}^nv(B_j)^p\leq\frac{1}{(40) ^p}\sum \limits _{j=k+1}^{n}v(A^{*}_j(\mathcal{J},p))^p 
\end{align}

We add $\sum \limits_{i=1}^k v \left( g_i\right)^p$ to both sides of the previous equation and apply Lemma~\ref{Induction_argument} (with $p \in (-\infty, 0)$)  to obtain the desired inequality 
\begin{align*}
\sum \limits_{i=1}^k v \left( g_i\right)^p + \sum \limits _{j=k+1}^n v(B_i)^p \leq \frac{1}{\left( 40 \right)^p}  \left( (40)^p\sum \limits _{i=1}^k v(g_i)^p + \sum \limits _{j=k+1}^{n} v(A^{*}_j(\mathcal{J},p))^p \right)
\leq \frac{1}{40 ^p}\left( \sum \limits _{i=1}^n v(A^{*}_i(\mathcal{I},p))^p\right)\hspace{3.6cm}
\end{align*}

\noindent {\bf Case 2:} $p  \in (0,0.4)$. 
Note that in this case $p>0$. Hence, to obtain the desired inequality, ${\rm M}_p(\mathcal{A})\geq \frac{1}{40} {\rm M}_p(\mathcal{A}^{*}(\mathcal{I},p))$, it suffices to show that
\begin{align*}
\sum \limits_{i=1}^k v \left( g_i\right)^p + \sum \limits_{j=k+1}^{n} v\left(B_j\right)^p\geq \frac{1}{(40) ^p}\sum \limits_{i=1}^n v(A^{*}_i(\mathcal{I},p))^p 
\end{align*}

Exponentiating both sides of equation (\ref{16}) by $p >0$ and summing over $j \in \{k+1, \ldots, n\}$ lead to 

\begin{align*}
\sum \limits _{j=k+1}^nv(B_i)^p & \geq \frac{1}{(40) ^p}\sum \limits _{j=k+1}^{n}v(A^{*}_j(\mathcal{J},p))^p.
\end{align*}

We add $\sum \limits_{i=1}^k v \left( g_i\right)^p$ to both sides of the previous equation and apply Lemma~\ref{Induction_argument} (with $p \in (0,0.4)$)  to obtain the desired inequality 
\begin{align*}
\sum \limits_{i=1}^k v \left( g_i\right)^p + \sum \limits _{j=k+1}^n v(B_j)^p &\geq  \frac{1}{\left( 40\right)^p}  \left( (40)^p \sum \limits _{i=1}^k v(g_i)^p + \sum \limits _{j=k+1}^{n} v(A^{*}_j(\mathcal{J},p))^p \right)
\geq \frac{1}{(40) ^p}\left( \sum \limits _{i=1}^n v(A^{*}_i(\mathcal{I},p))^p\right)\hspace{3.6cm}
\end{align*}

\noindent {\bf Case 3:} $p=0$. In this case, ${\rm M}_0(\mathcal{A}) = \left( \prod \limits _{i=1}^kv(g_i) \prod \limits _{j=k+1}^{n} v(B_j)\right)^\frac{1}{n}.$
Multiplying inequality (\ref{16}) over all $j \in \{k+1,\ldots, n\}$ gives us 
\begin{align*}
\prod \limits _{j=k+1}^{n} v(B_i)  \geq \frac{1}{(40)^{n-k}} \left( {\rm M}_0 \left(\mathcal{A}^{*}(\mathcal{J}, 0) \right) \right)^{n-k} & =  \frac{1}{(40)^{n-k}}\prod \limits _{j=k+1}^{n} v(A^{*}_j(\mathcal{J},0))
\end{align*}

Next we multiply both sides of this inequality by  $\prod_{i=1}^k v\left(g_i\right)$ and obtain 
\begin{align*} 
\prod \limits _{i=1}^k v(g_i) \prod \limits _{j=k+1}^n v(B_i)\geq \frac{1}{(40) ^{n-k}}\prod \limits_{i=1}^k v(g_i) &\prod \limits_{j=k+1}^{n} v(A^*_j (\mathcal{J},0) ) \geq \frac{1}{(40) ^n}({\rm M}_0(\mathcal{A}^{*}(\mathcal{I}, 0)))^n   \tag{via Lemma~\ref{Induction_argument} with $p=0$}
\end{align*}
Taking the $n${th} root on both sides, we obtain the desired result ${\rm M}_0(\mathcal{A}) \geq \frac{1}{40} {\rm M}_0(\mathcal{A}^{*}(\mathcal{I},0))$. \\

Theorem \ref{MainTheorem} now stands proved for $p \in (-\infty,0.4)$.

We next complete the main result by addressing the range $p \in [0.4, 1]$. It is relevant to note that the arguments used to establish the approximation guarantee in this range are different from the techniques used so far. 

\section{Proof of Theorem \ref{MainTheorem} for $p \in [0.4,1]$} \label{subsection:p-half-one}



For instance $\mathcal{I}$, let \textsc{Alg} assign the $k$ highest-valued goods as singletons in its while-loop. Specifically, write $\widehat{G}=\left\{g_1,\ldots, g_k\right\}$ to denote the $k$ goods that are assigned in {while}-loop of \textsc{Alg}. Instance $\mathcal{J} = \langle [m] \setminus \{g_1, \ldots, g_k\}, [n] \setminus [k], v \rangle$ is passed as input to \textsc{AlgLow}, which returns allocation $\mathcal{B}=(B_{k+1},\ldots,B_n)$. Recall that $\mathcal{B}$  satisfies Lemma~\ref{Low_valued}. Finally, let $\mathcal{A}=\Pallo$ denote the allocation returned by \textsc{Alg}.
Also, as before, let $\mathcal{A}^*(\mathcal{I},p) = (A^*_1(\mathcal{I},p),\ldots, A^*_n(\mathcal{I},p))$ denote the $p$-optimal allocation of $\mathcal{I}$.

We will prove the following bound for $p \in  [0.4,1]$ and, hence, establish the stated approximation guarantee 
\begin{align}
\left( \frac{1}{n}\sum\limits_{i=1}^k v(g_i)^p +  \frac{1}{n}\sum\limits_{j=k+1}^{n}v(B_j)^p \right)^{\frac{1}{p}} & \geq \frac{1}{40} \left(\frac{1}{n} \sum\limits_{i=1}^{n}v(A^{*}_i(\mathcal{I},p))^p \right)^{\frac{1}{p}}
\end{align}

Write $\mathcal{A}^*(\mathcal{I},p)\setminus \widehat{G}$ to denote the allocation (specifically, an $n$-partition) obtained by removing the goods $\widehat{G} =\{g_1,\ldots,g_k\}$ from the bundles in $\mathcal{A}^*(\mathcal{I},p)$, i.e., $ \mathcal{A}^*(\mathcal{I},p)\setminus \widehat{G} \coloneqq \left( A^*_i(\mathcal{I},p) \setminus \widehat{G}  \right)_{i=1}^n$. 
Subadditivity of $v$ ensures that, for all $i \in [n]$, the bundle $A^{*}_i(\mathcal{I},p)$ satisfies  
\begin{align*}
v(A^{*}_i(\mathcal{I},p))  & \leq v \left( A^{*}_i (\mathcal{I},p)  \setminus \widehat{G}\right)  + \sum \limits_{ g \in \widehat{G} \cap A^{*}_i(\mathcal{I},p)} v(g).  
\end{align*}

Since $p>0$, exponentiating the previous inequality by $p$ gives us   
\begin{align}
  \left(v(A^{*}_i(\mathcal{I},p)) \right)^p & \leq \left(v \left( A^{*}_i (\mathcal{I},p)  \setminus \widehat{G}\right)  + \sum \limits_{ g \in \widehat{G} \cap A^{*}_i(\mathcal{I},p)} v(g) \right)^p \nonumber \\
  & \leq \left(v(A^{*}_i(\mathcal{I},p)\setminus \widehat{G}) \right)^p + \sum \limits_{g\in \widehat{G}\cap A^{*}_i(\mathcal{I},p)} v(g)^p \label{ineq:p-exp}
\end{align}

The last inequality follows from the fact that $(x+y)^p \le x^p+y^p$, for all $p \in [0.4, 1]$ and $x, y \in \mathbb{R}_+$.

Averaging equation (\ref{ineq:p-exp}) over $i \in [n]$ leads to  
 \begin{align}\label{17}
  \frac{1}{n} \summi{1}{k}{g}^p + \frac{1}{n}\sum\limits_{i=1}^{n}{v(A^{*}_i(\mathcal{I},p)\setminus \widehat{G}))^p} 
  \geq \frac{1}{n}\sum\limits_{i=1}^{n}{v(A^{*}_i(\mathcal{I},p))^p}
 \end{align}

To show that the $k$ goods in $\widehat{G}$ (which are allocated as singletons) substantially contribute towards $p$-mean welfare of the computed allocation $\mathcal{A}$, we will next establish the following lower bound for all $1\leq t \leq k$
\begin{align*}
v(g_t)^p & \geq \frac{1}{(7.06)^p}\left( \frac{1}{n}\sum\limits_{j=1}^n v(A^{*}_j(\mathcal{I},p)\setminus \widehat{G}) ^p \right)
\end{align*}

Recall that $\mathcal{I}^t$ denotes the fair-division instance $\I{[m]\setminus \{ g_1,\ldots, g_t\}}{[n]\setminus \{1,\ldots t\}}$, for $1\leq t\leq k$. Write $\mathcal{A}^*(\mathcal{I}^t,p)=(A^*_{t+1} (\mathcal{I}^t,p),\ldots,A^*_{n}(\mathcal{I}^t,p))$ to denote a $p$-optimal allocation of instance $\mathcal{I}^t$. 

The selection criterion of the while-loop in \textsc{Alg} and the fact that Feige's algorithm achieves an approximation ratio of $2$ ensure 
\begin{align}
  v(g_t) &\geq \frac{1}{3.53}{\rm F}(\mathcal{I}^{t-1}) \geq \frac{1}{7.06}\left ( \frac{1}{n-t+1} \sum\limits_{j=t}^{n}v(A^{*}_j(\mathcal{I}^{t-1},1))  \right) \label{ineq:bound-g-t-low}
\end{align}

Index the $n$ bundles in allocation $\mathcal{A}^*(\mathcal{I},p)\setminus \widehat{G}$ in non-increasing order of value $  v(A^*_1(\mathcal{I},p)\setminus \widehat{G}) \geq  v(A^*_2(\mathcal{I},p)\setminus \widehat{G}) \geq \ldots \geq v(A^*_n(\mathcal{I},p)\setminus \widehat{G})$ and note that the arithmetic mean of the values of the first $n-t+1$ bundles is at least as large as the overall arithmetic mean
\begin{align}
\frac{1}{n-t+1}\sum \limits _{j=1}^{n-t+1}v(A^*_j(\mathcal{I},p)\setminus \widehat{G}) \geq \frac{1}{n}\sum \limits _{j=1}^{n}v(A^*_j(\mathcal{I},p)\setminus \widehat{G}) \label{ineq:top-mean}
\end{align}

Given that allocation $\mathcal{A}^*(\mathcal{I},p)\setminus \widehat{G}$ constitutes an $n$-partition of the set of goods $[m] \setminus \widehat{G}$ and allocation $\mathcal{A}^*(\mathcal{I}^{t-1}, 1) = (A^*_{t} (\mathcal{I}^{t-1},1),\ldots, A^*_{n}(\mathcal{I}^{t-1},1))$ is an $(n-t+1)$-partition of $[m]\setminus \{g_1, \ldots, g_{t-1}\} \supseteq [m] \setminus \widehat{G}$, we have the following containment  $\bigcup \limits _{j=1}^{n-t+1}\left(A_j^*(\mathcal{I},p)\setminus \widehat{G}\right ) \subseteq \bigcup \limits _{j=t}^{n} \left(\ A^*_j(\mathcal{I}^{t-1},1)\right)$. Furthermore, by definition, allocation $\mathcal{A}^*(\mathcal{I}^{t-1}, 1)$ achieves the maximum possible average social welfare among all $(n-t+1)$ partitions of $[m] \setminus \{g_1, \ldots, g_t\}$. Therefore,  we have 
\begin{align}
 \frac{1}{n-t+1}\sum\limits_{j=t}^{n} v(A^{*}_j(\mathcal{I}^{t-1},1)) & \geq   \frac{1}{n-t+1}\sum\limits_{j=1}^{n-t+1} v(A^{*}_j(\mathcal{I},p)\setminus \widehat{G}) \label{ineq:opt-sub-problem}
\end{align}

Equations (\ref{ineq:bound-g-t-low}) and (\ref{ineq:opt-sub-problem}) lead to 
 \begin{align*}
  v(g_t) & \geq \frac{1}{7.06}\left( \frac{1}{n-t+1}\sum\limits_{j=1}^{n-t+1} v(A^{*}_j(\mathcal{I},p)\setminus \widehat{G})  \right)\\
  &\geq \frac{1}{7.06}\left( \frac{1}{n}\sum\limits_{j=1}^n v(A^{*}_j(\mathcal{I},p)\setminus \widehat{G})  \right) \tag{using inequality (\ref{ineq:top-mean})} \\ 
  &\geq \frac{1}{7.06}\left( \frac{1}{n}\sum\limits_{j=1}^n v(A^{*}_j(\mathcal{I},p)\setminus \widehat{G})^p  \right) ^{\frac{1}{p}} \tag{via the generalized mean inequality}
\end{align*}
Exponentiating both sides of the previous inequality by $p$ gives us the desired lower bound
\begin{align}\label{19}
 v(g_t)^p \ge \frac{1}{(7.06)^p}\left( \frac{1}{n}\sum\limits_{j=1}^n v(A^{*}_j(\mathcal{I},p)\setminus \widehat{G}) ^p \right)
 \end{align} 
 
Equation (\ref{19}) enables us to bound the $p$-welfare contribution of the goods $\widehat{G}=\{g_1, \ldots, g_k\}$ assigned as singletons 
\begin{align}\label{20}
 \frac{2}{n} \sum\limits_{i=1}^kv(g_i)^p \geq  \frac{1}{n} \sum\limits_{i=1}^kv(g_i)^p
 + \frac{k}{n} \ \frac{1}{(7.06)^p} \left(\frac{1}{n} \sum \limits_{i=1}^n v ( A^{*}_i (\mathcal{I},p)\setminus \widehat{G})^p\right)
\end{align}

Recall that \textsc{AlgLow}---with input instance $\mathcal{J}=\mathcal{I}^k$---returns allocation $\mathcal{B}=(B_{k+1}, \ldots B_n)$. Next we lower bound the values of these bundles $B_{k+1}, \ldots, B_n$. 
\begin{align*}
v(B_j) & \geq  \frac{1}{40} {\rm M}_1(\mathcal{A}^{*}(\mathcal{J},1)) \tag{via Lemma~\ref{Low_valued}} \\
& = \frac{1}{40 }\left( \frac{1}{n-k} \sum\limits_{i=k+1}^{n} v(A^{*}_i(\mathcal{J},1))  \right) \tag{by defintion, $\mathcal{A}^{*}(\mathcal{J},1) = \left( \ A^{*}_i(\mathcal{J},1) \ \right)_{i=k+1}^n $} \\
& \geq \frac{1}{40 }\left(  \frac{1}{n-k}\sum\limits_{j=1}^{n-k} v(A^{*}_j(\mathcal{I},p)\setminus \widehat{G})  \right)\tag{instantiating inequality (\ref{ineq:opt-sub-problem}) for $\mathcal{J} = \mathcal{I}^k$} \\
&\geq \frac{1}{40}\left( \frac{1}{n}\sum\limits_{j=1}^n v(A^{*}_j(\mathcal{I},p)\setminus \widehat{G})  \right) \tag{using inequality (\ref{ineq:top-mean})} \\ 
  &\geq \frac{1}{40}\left( \frac{1}{n}\sum\limits_{j=1}^n v(A^{*}_j(\mathcal{I},p)\setminus \widehat{G})^p  \right) ^{\frac{1}{p}} \tag{via the generalized mean inequality}
\end{align*}

Exponentiating by $p$ and summing over all $j \in \{k+1,\ldots n\}$, we have
\begin{align} \label{21}
 \frac{1}{n}\sum\limits_{j=k+1}^{n}v(B_j)^p \ge \frac{n-k}{n} \frac{1}{(40) ^p}\left(\frac{1}{n}
 \sum\limits_{i=1}^{n}v(A^{*}_i(\mathcal{I},p)\setminus \widehat{G})^p \right)
\end{align}

Combining inequalities (\ref{20}) and (\ref{21}) gives us
\begin{align*}
 \frac{1}{n}\sum\limits_{i=1}^k v(g_i)^p +  \frac{1}{n}\sum\limits_{j=k+1}^{n}v(B_j)^p & \geq
   \frac{1}{2n}\ \sum\limits_{i=1}^k v(g_i)^p +
  \frac{k}{2n} \cdot \frac{1}{(7.06)^p} \left (\frac{1}{n}\sum\limits_{j=1}^{n}v ( A^{*}_i (\mathcal{I},p)\setminus \widehat{G})^p\right) \\ & \ \ \ \ \   +
  \frac{n-k}{n}\frac{1}{(40) ^p}\left(\frac{1}{n}\sum\limits_{i=1}^{n}v(A^{*}_i(\mathcal{I},p)\setminus \widehat{G})^p\right) 
  \end{align*}

Note that $2 \times (7.06)^p \leq (40) ^p$ for all $p  \in [0.4,1]$, hence, the previous inequality simplifies to 
  \begin{align*}
    \frac{1}{n}\sum\limits_{i=1}^k v(g_i)^p +  \frac{1}{n}\sum\limits_{j=k+1}^{n}v(B_j)^p & \geq
    \frac{1}{2n}\sum\limits_{i=1}^{k}v(g_i)^p +
  \frac{1}{(40) ^p}\left(\frac{1}{n}\sum\limits_{i=1}^{n}v(A^{*}_i(\mathcal{I},p)\setminus \widehat{G})^p\right) \\
  & \geq \frac{1}{(40) ^p} \left(\frac{1}{n} \sum\limits_{i=1}^{k}v(g_i)^p +
 \frac{1}{n} \sum\limits_{i=1}^{n}v(A^{*}_i(\mathcal{I},p)\setminus \widehat{G})^p \right) \tag{since $(40)^p > 2$ for $p \in [0.4,1]$} \\
 & \geq \frac{1}{(40) ^p}\left(\frac{1}{n} \sum\limits_{i=1}^{n}v(A^{*}_i(\mathcal{I},p))^p \right) \tag{using inequality (\ref{17})}
  \end{align*}

Taking the $p${th} root (with $p>0$) on both sides of the last inequality gives us the desired result for the computed allocation $\mathcal{A} = \Pallo$
\begin{align*}
{\rm M}_p(\mathcal{A})\ge \frac{1}{40}{\rm M}_p(\mathcal{A}^{*}(\mathcal{I},p)).
\end{align*}

This completes the proof of Theorem \ref{MainTheorem} for all $p \in (-\infty,1]$.

\section{Conclusion and Future Work}
This work studies the problem of allocating indivisible goods among agents that share a common subadditive valuation. We show that, for such settings, one can always (and in polynomial-time) find a single allocation that simultaneously approximates a range of generalized-mean welfares, to within a constant factor of the optimal. 


For ease of presentation, we focussed on the case in which the agents' valuations are exactly identical. Nonetheless, it can be shown that the developed results are somewhat robust: if, say, the agents' valuations are point-wise and multiplicatively close to each other, then again one can obtain meaningful approximation guarantees. Here, an interesting direction of future work is to address settings in which we have a fixed number of distinct valuation functions across all the agents. A nontrivial improvement on the developed approximation guarantee will also be interesting.

\section*{Acknowledgements}
Siddharth Barman gratefully acknowledges the support of a Ramanujan Fellowship (SERB - {SB/S2/RJN-128/2015}) and a Pratiksha Trust Young Investigator Award.

\bibliographystyle{alpha}
\bibliography{ref}

\appendix

\section{Proof of Numeric Inequality from Lemma \ref{Good_Transfer}} \label{app 3}
This section establishes the numeric inequalities used in Sections \ref{subsection:p-infty-zero-good-transfer} and \ref{subsection:p-zero-half-good-transfer}. Specifically, for $p\in (-\infty,0)$
\begin{align}
 \left(  \frac{1}{2} - \frac{1}{40}\right )^p + \left ( \frac{1}{2}\right )^p &\leq 1 + \left ( \frac{2}{11.33} \right )^p \label{ineq_neg}
\end{align}
Also, for $p\in (0,0.4)$, we have 
\begin{align}
\left(  \frac{1}{2} - \frac{1}{40}\right )^p + \left ( \frac{1}{2}\right )^p &\geq 1 + \left ( \frac{2}{11.33} \right )^p \label{ineq_pos} 
\end{align}

Write $a$ to denote $\left(  \frac{1}{2} - \frac{1}{40}\right ) $, $b$ to denote $\left ( \frac{1}{2}\right )$ and $c$ to denote $\left ( \frac{2}{11.33} \right ).$ Consider the following differentiable function $f(p)\coloneqq a^p+b^p-1-c^p$. In order to obtain bounds (\ref{ineq_neg}) and (\ref{ineq_pos}), it suffices to prove that $f(p)>0$ for all $p\in (0,0.4)$ and $f(p)\leq 0$ for all $p\in(-\infty,0)$. We first show a useful property of the function $f(p)$ that will help us in deriving these inequalities.
\begin{center}
    \includegraphics[scale=0.3]{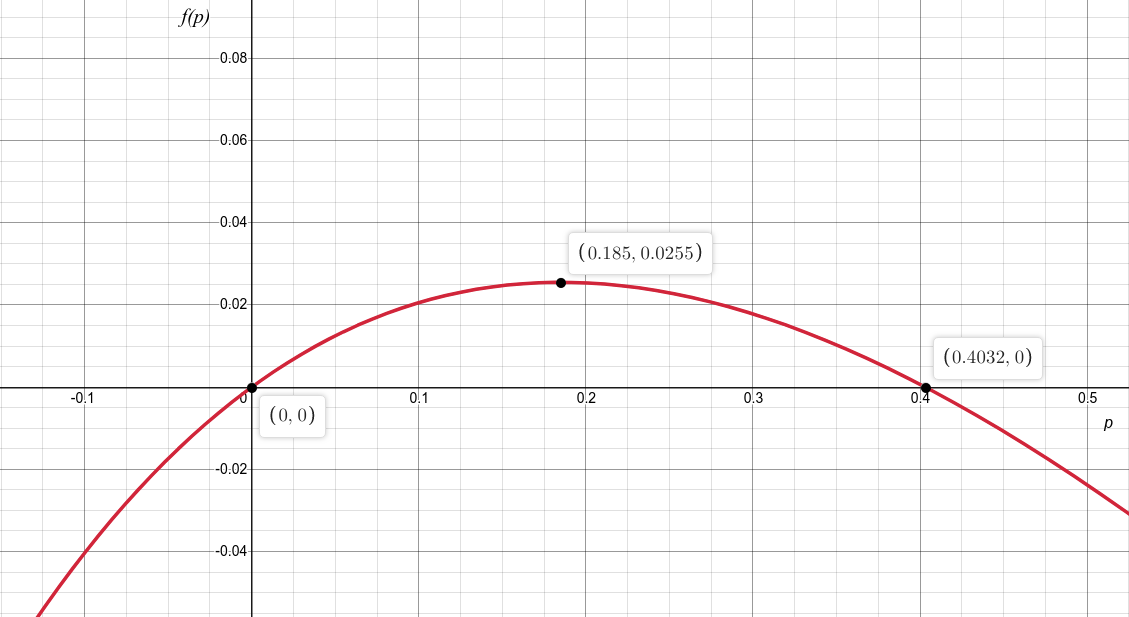}
    \captionof{figure}{Plot of the function  $f(p)$ }
    \label{fig:sample_figure}
\end{center}

\begin{Claim}\label{extreme_points_are_maxima}
Every extreme point of the function $f(p)=a^p+b^p-1-c^p$ is a local maximum; here, $a=\left(  \frac{1}{2} - \frac{1}{40}\right ) $, $b=\left ( \frac{1}{2}\right )$ and $c=\left ( \frac{2}{11.33} \right ).$
\end{Claim}

\begin{proof}
We consider the first and second derivatives of the smooth function $f$:
\begin{align}
f'(p)= a^p\log  (a) + b^p\log  (b) - c^p\log  (c) \label{$f'$}\\  \text{and } f''(p) = a^p(\log (a))^2 + b^p(\log  (b))^2 - c^p(\log  (c))^2  \label{$f''$}
\end{align} 
Let $\overline{p}$ be any point that satisfies $f'(\overline{p})=0. $ That is, $c^{\overline{p}}\log (c)= a^{\overline{p}}\log (a) + b^{\overline{p}}\log  (b) $. Instantiating equation (\ref{$f''$}) with $p=\overline{p}$ and substituting the previous expression for $c^{\overline{p}}\log (c)$, we get 
\begin{center}
$f''(\overline{p})=a^{\overline{p}}\log  (a)(\log  (a)-\log  (c))+b^{\overline{p}}\log  (b)(\log (b)-\log (c))$
\end{center}
In the above equation, note that the terms $(\log  (a)-\log  (c))$ and $(\log  (b)-\log  (c))$ are positive, since $\frac{a}{c} >1$ and $\frac{b}{c} >1$, while the terms $\log (a)$ and $\log (b)$ are negative since $a,b<1.$ Therefore, $f''(\overline{p})<0$, which means $\overline{p}$ is a local maximum. We have now shown that every extreme point of $f$ is a local maximum.

\end{proof}
Note that $f(0.4)>0$ and $f(0.41)<0$, hence, the intermediate value theorem implies that there exists a point $r \in (0.4, 0.41)$ such that $f(r) = 0$. We will show that $f$ is nonnegative over the interval $[0,r]$ and, hence, $f(p) \geq 0$ for all $p \in (0, 0.4)$. 

Given that $f(0)=0$ and $f(r) = 0$, we can apply the mean value theorem to $f$ in the range $[0,r]$ to conclude that there is a point $p_0 \in [0,r]$ satisfying $f'(p_0)=0$. Using Claim \ref{extreme_points_are_maxima}, we will next show that $f(p)$ has a unique extreme point, which is a local maximum, in the range $[0,r]$. 

\begin{Claim}\label{unique_maximum}
Let $p^*  \in [0,r]$ be a point that satisfies $f'(p^*)=0$, then it is the unique extreme point of $f(\cdot)$ and is a local maximum.
\end{Claim}
\begin{proof}
Since $f'(p^*)=0$, $p^*$ is an extreme point. Additionally, from Claim \ref{extreme_points_are_maxima}, $p^*$ is a local maximum. 

Assume, towards a contradiction, that there is another local maximum $p_1 \in \mathbb{R}$ and, without loss of generality, assume that $p_1>p^*.$ We will show that this implies the existence of a local minimum between $p^*$ and $p_1$, hence contradicting Claim \ref{extreme_points_are_maxima}.

Consider the point $p_{min}$ defined as $p_{min}=\text{arginf} \left\{ \ f(p) \mid p \in [p^*, p_1] \right\}$. Note that $p_{min}\neq p^*,p_1$ because $p^*$ and $p_1$ are local maxima, which means their function values are greater than those of the points in an $\varepsilon$-neighborhood around them. Therefore, applying Fermat's theorem for stationary points, we get $f'(p_{min})=0$. Since $p_{min}$ is also the infimum of $p\in[p^*,p_1],$ it is a local minimum  between $p^*$ and $p_1.$ Therefore, by way of contradiction, we obtain the stated claim. 
\end{proof}
\noindent \textit{Proof of inequality (\ref{ineq_pos}):} We now prove that $f(p)\geq0$ for $p \in [0, 0.4].$ 
Assume, for the sake of contradiction, that there exists is a point $q \in [0, 0.4]$ such that $f(q)<0$.

The facts that $f(0) = 0$, $f'(0) >0$ and $f(q)<0$, along with the mean value theorem, imply that there exists a point $\widehat{p} \in [0, q)$ such that $f'(\widehat{p})=0$. Additionally, applying mean value theorem with inequalities $f(0.4)>0$, $f'(0.4)<0$, and $f(q)<0$, we get a different point $\overline{p}\in (q, 0.4]$ such that $f'(\overline{p})=0$. Existence of two distinct extreme points contradicts Claim \ref{unique_maximum} and, hence, establishes that $f(p)\geq 0$ for all $p \in [0,0.4]$. \\

\noindent \textit{Proof of inequality (\ref{ineq_neg}):} We split the proof of this inequality depending on the range of $p$:\\

\noindent
Case 1: $p\in (-\infty,-1]$. Note that $a>\frac{1}{2.2}$, $b=\frac{1}{2}$, $c<\frac{1}{5}$, and $p$ is negative.
Substituting these bounds for $a,b$ and $c$ in the expression for $f(p)$, we get $f(p)<\left(\frac{1}{2.2}\right)^p+\left(\frac{1}{2}\right)^p-1-\left(\frac{1}{5}\right)^p.$  Equivalently, we have $f(p)<(2.2)^{|p|}+2^{|p|}-1-5^{|p|}.$  Since $5^{|p|}\geq(2.2)^{|p|}+2^{|p|}$ for $p\leq-1$, we conclude that $f(p)<0$ for $p\in(-\infty,-1].$\\
\noindent
Case 2: $p\in (-1,0)$. Assume, for the sake of contradiction, that there is a $\widetilde{p} \in(-1,0)$ such that $f(\widetilde{p})>0$.
Note that $f(-1) <0$. Hence, by applying intermediate value theorem to $-1$ and $\widetilde{p}$, we have a point $\widehat{p} \in (-1, \widetilde{p})$ such that $f(\widehat{p})=0$. Additionally, mean value theorem with $0$ and $\widehat{p}$ implies that there exists a point $\overline{p}\in [\widehat{p},0)$ such that $f'(\overline{p})=0$. Since this contradicts Claim \ref{unique_maximum}, we have $f(p) \leq 0$ for all $p\in(-1,0)$.

These two cases imply that inequality (\ref{ineq_neg}) holds for all $p\in (-\infty,0).$

\section{APX-Hardness of Maximizing $p$-Mean Welfare} \label{APX_Hardness}
In this section, we prove that the problem of computing a $p$-optimal allocation is {\rm APX}-hard, for all $p\in (-\infty,1]$, in the demand oracle model. 
That is, we show that there exists a constant $c>1$ such that it is {\rm NP}-hard to approximate the optimal $p$-mean welfare within a factor of $c$, even if we are given access to demand queries.

\begin{theorem}\label{APX_hardness_theorem}
Given a fair division instance $\mathcal{I}=\I{[m]}{[n]}$, wherein the agents have identical, subadditive valuations, the $p$-mean welfare maximization problem is {\rm APX}-Hard for all $p\in (-\infty,1]$, in the demand oracle model.
 \end{theorem}

We prove this hardness result by developing a gap-preserving reduction from the Gap-3DM problem. An instance $\mathcal{C}$ of this problem consists of three disjoint sets $X$, $Y$, and $Z$, of cardinality $q$ each, along with a collection of $3$-uniform hyperedges $E \subset X \times Y \times Z$. 
The goal is to find a matching (i.e., a subset of pairwise disjoint hyperedges) $M \subset E$ of maximum cardinality. 

Formally, Gap-3DM is the gap version of $3$-dimensional matching and it entails distinguishing between the following types of instances: \\
(i) {\rm YES}-instance: There is a perfect matching (a matching of size $q$) in the given instance $\mathcal{C}$.\\
(ii) {\rm NO}-instance: All matchings in $\mathcal{C}$ are of size at most $\alpha q$, with $\alpha <1$.\\
The Gap-3DM problem is known to be {\rm NP}-hard, for an absolute constant $\alpha <1$~\cite{ostrovsky2014s}.\footnote{Note that the given instances in this  problem are promised to be either {\rm YES} or {\rm NO} instances.} 

Theorem \ref{APX_hardness_theorem} follows from a gap-preserving reduction. In particular, we will prove that in the {\rm YES} case (i.e., when 
the given instance of Gap-3DM has a perfect matching) the reduced instance of the welfare problem admits an allocation with $p$-mean welfare at least $3$. In the {\rm NO} case, the optimal $p$-mean welfare is less than $3 \ c (\alpha)$, where $c (\alpha) <1$ is a constant (that depends of $\alpha$). 

Therefore, given a $\frac{1}{c(\alpha)}$-approximation algorithm for $p$-mean welfare maximization, one could distinguish between the {\rm YES} and {\rm NO} instances of Gap-3DM. Hence, the hardness of Gap-3DM implies that it is {\rm NP}-hard to approximate the optimal $p$-mean welfare with a constant factor of $\nicefrac{1}{c(\alpha)}$. 

\subsection{Proof of Theorem \ref{APX_hardness_theorem}} 

Given an instance $\mathcal{C}$ of Gap-3DM with $3$-uniform hyperedges $E=\{E_1, \ldots, E_T\}$ over size-$q$ sets $X$, $Y$, and $Z$, we construct an instance $\mathcal{I}$ of $p$-mean welfare maximization with $q$ agents and $3q$ goods, one for each vertex $X \cup Y \cup Z$. All the agents share a common XOS valuation function $v$: define $v(S) \coloneqq \max \limits _{1\leq i\leq T} \left\{|S \cap E_i| \right\}$ for each subset of goods $S\subset X\cup Y\cup Z$.

Note that the value of any subset $S$ is upper bounded by 3, since each $E_i$ is a 3-uniform hyperedge. We now prove that this reduction is gap preserving. 

When the input instance $\mathcal{C}$ is a {\rm YES} instance, there is a matching of size $q$. We assign the three goods corresponding to each edge in the matching to a distinct agent. In this allocation each agent gets a bundle of value $3$. Therefore, the optimal $p$-mean welfare in this case is at least $3$. Conversely, when $\mathcal{C}$ is a {\rm NO} instance, every matching is of size at most $\alpha q$. We claim that in this case, the optimal $p$-mean welfare in the reduced instance $\mathcal{I}$ is upper bounded by $3 \ c(\alpha)$, where $c(\alpha) =\frac{2+\alpha}{3} <1$. 

Let $\mathcal{A}^*(\mathcal{I},p)=(A^*_1,\ldots, A^*_q)$ denote a $p$-optimal allocation in instance $\mathcal{I}$. Here, in any allocation, an agent's value for her bundle is either $0$, $1$, $2$, or $3$. We partition the bundle in the optimal allocation into two collections
$\mathcal{H} \coloneqq \{A^*_i \mid v(A^*_i)=3\}  \text{ and } \overline{\mathcal{H}} \coloneqq \{A^*_i \mid v(A^*_i)\leq 2\}$.
Every bundle $A^*_i$ in $\mathcal{H}$ contains at least one hyperedge $E_j$. Since the bundles in $\mathcal{H}$ are disjoint, the hyperedges contained in different bundles are nonintersecting. Hence, such hyperedges form a matching in $\mathcal{C}$ of size at least $|\mathcal{H}|.$ Recall that $\mathcal{C}$ is a {\rm NO} instance, i.e., every matching in $\mathcal{C}$ of size at most $\alpha q.$ Therefore,  there are at most $\alpha q$ bundles in $\mathcal{A}^*(\mathcal{I},p)$ of value $3$, $|\mathcal{H}| \leq \alpha q$. Using this inequality we can upper bound the optimal $p$-mean welfare in instance $\mathcal{I}$ as follows. 
\begin{Claim}
${\rm M}_p(\mathcal{A}^*(\mathcal{I},p))\leq 2+ \alpha = 3 \left( \frac{2+\alpha}{3}\right)  $
\end{Claim}
\begin{proof}
Recall that $\mathcal{H}$ and $\overline{\mathcal{H}}$ form a partition of the $q$ bundles in $\mathcal{A}^*(\mathcal{I},p)$. Let $|\mathcal{H}|=\overline{\alpha} q$ for some $\overline{\alpha}\leq \alpha$. Then, $\sum \limits _{A^*_i \in \overline{\mathcal{H}}} v(A^*_i)\leq 2(1-\overline{\alpha})q$. Therefore, 
\begin{align*}
{\rm M}_p(\mathcal{A}^*(\mathcal{I},p)) &\leq {\rm M}_1(\mathcal{A}^*(\mathcal{I},p)) \tag{via the generalized mean inequality}\\
& \leq \frac{1}{q}\left( 3\overline{\alpha} q + 2(1-\overline{\alpha})q \right) \tag{averaging over the values of bundles in $\mathcal{H}$ and $\overline{\mathcal{H}}$}\\
&\leq 2+\alpha \tag{since $\overline{\alpha}\leq \alpha$}
\end{align*}
\end{proof}
Hence, we have a polynomial-time reduction from the Gap-3DM to the $p$-mean welfare maximization problem such that:\\
(i) When Gap-3DM instance $\mathcal{C}$ is a {\rm YES} instance, the optimal $p$-mean welfare is at least $3$.\\
(ii) When $\mathcal{C}$ is a {\rm NO} instance, the optimal $p$-mean welfare is at most $3\ c(\alpha)$, for a constant $c(\alpha) <1$.  

As mentioned previously, such a gap-preserving reduction establishes the {\rm APX}-hardness of $p$-mean welfare maximization. 

Finally, note that we can efficiently simulate the demand oracle for the valuation function $v$ in the constructed instance $\mathcal{I}$. For any additive function $f$, the response to a demand query---with prices $p_j$s associated with the goods---is simply the subset of goods $g$ that satisfy $f(g) - p_g \geq 0$. The XOS valuation function $v$ in the reduction is obtained by considering a maximum over $T$ additive functions. The parameter $T$ is the number of edges in the given Gap-3DM instance and, hence, is polynomially bounded. Therefore, we can efficiently simulate the demand oracle for $v$ by explicitly optimizing over all the $T$ additive functions.  Overall, we get that the {\rm APX}-hardness holds in the demand oracle model and the theorem follows. 

\end{document}